\PassOptionsToPackage{table}{xcolor}
\PassOptionsToPackage{dvipsnames}{xcolor}
\documentclass[acmlarge,authorversion,screen,nonacm]{acmart}
\makeatletter
\let\@authorsaddresses\@empty
\makeatother

\usepackage{braket}
\usepackage{amsthm}
\usepackage{thm-restate} 
\usepackage{mathtools}
\usepackage{hyperref}
\usepackage[capitalize,nameinlink,noabbrev]{cleveref}
\usepackage{crossreftools}
\usepackage[utf8]{inputenc}
\usepackage{xcolor}
\usepackage{microtype}
\usepackage{comment}
\usepackage{xspace}
\usepackage{stmaryrd}

\newcommand{\dra}[1]{\llbracket #1 |}

\usepackage[xcolor,hyperref,notion,electronic,makeidx]{knowledge} 
\knowledge{notion}
    | choice function
\knowledge{notion}
    | Single choice
    | single choice
    | single choice condition
\knowledge{notion}
    | Multiple choice
    | multiple choice
    | multi-choice
    | multiple choice condition
    | multi-choice condition
\knowledge{notion}
    | Layered choice
    | layered choice
    | layered condition
    | layered choice condition
\knowledge{notion}
    | Injective layered choice
    | injective layered choice
    | INJECTIVE LAYERED CHOICE
    | injective layered condition
    | injective layered choice condition
    
\knowledge{notion}
    | structure
    | structrures
    | relational structure
    | relational structures
    | relational templates
\knowledge{notion}
    | similar
\knowledge{notion}
    | homomorphism
    | homomorphically
\knowledge{notion}
    | PCSP template
    | template@PCSP
    | templates@PCSP
\knowledge{notion}
    | polymorphism
    | polymorphisms
    | compatible with
    | compatibility with
\knowledge{notion}
    | minion
    | minions
    | (function) minion
\knowledge{notion}
    | minor
    | minors
\knowledge{notion}
    | minor map
    | minor maps
\knowledge{notion}
    | projection
    | projections
    | dictator
\knowledge{notion}
    | idempotent
    | idempotency
\knowledge{notion}
    | essential
    | non-essential
    | essentiality
\knowledge{notion}
    | minion generated by
\knowledge{notion}
    | folded
    | foldedness
\knowledge{notion}
    | monotone
    | monotonicity
    | antimonotone
\knowledge{notion}
    | eithertone
    
\knowledge{notion}
    | represent
    | representation
    | representations
\knowledge{notion}
    | fixing
    | fixing pair
    | fixing pairs
\knowledge{notion}
    | c-minion
    | c-minions
    | minion@c
    | minions@c
\knowledge{notion}
    | minor@c
    | minors@c
\knowledge{notion}
    | constant@c

\knowledge{notion}
    | injective
    | injective on
\knowledge{notion}
    | covered by
\knowledge{notion}
    | dominating set
    | dominating
    | dominates

\knowledge{notion}
    | instance@LC
    | instances@LC
    | variable@LC
    | variables@LC
    | vertex@LC
    | labels@LC
    | domain@LC
    | edge@LC
    | edges@LC
    | constraint@LC
    | constraints@LC
\knowledge{notion}
    | assignment@LC
    | solution@LC
    | satisfies@LC
    | satisfied@LC
    | satisfying@LC
    | satisfiable@LC
\knowledge{notion}
    | smooth@LC
    | smoothness@LC

\knowledge{notion}
    | instance@LLC
    | instances@LLC
    | variable@LLC
    | variables@LLC
    | layers@LLC
    | layered@LLC
    | Layered@LLC
    | labels@LLC
    | domain@LLC
    | edge@LLC
    | edges@LLC
    | constraint@LLC
    | constraints@LLC
\knowledge{notion}
    | transitive@LLC
    | transitivity@LLC
\knowledge{notion}
    | assignment@LLC
\knowledge{notion}
    | chain
    | chains
    | weakly satisfy
    | weakly satisfies
    | weakly satisfied
\knowledge{notion}
    | smooth@LLC
    | smoothness@LLC

\knowledge{notion}
    | minor condition
    | minor conditions
\knowledge{notion}
    | satisfy@MC
    | satisfied@MC
    | satisfies@MC
    | satisfiable@MC
    | assignment@MC
\knowledge{notion}
    | trivial@MC
    | triviality@MC

\knowledge{notion}
    | negated
\knowledge{notion}
    | power

\hypersetup{
  colorlinks=true
}

\DeclareUnicodeCharacter{0301}{\'a}

\DeclareMathOperator{\CSP}{CSP}
\DeclareMathOperator{\PCSP}{PCSP}
\DeclareMathOperator{\Pol}{Pol}
\DeclareMathOperator{\proj}{proj}

\knowledgenewrobustcmd{\NAE}[1]{\cmdkl{\text{NAE}_{#1}}}
\newrobustcmd{\kinl}[2]{#1\text{in}#2}
\knowledgenewrobustcmd{\PCondition}{\cmdkl{\mathrm{LOW}_{\bra{\cdot}}}}

\newcommand{\minion}[1]{\mathcal{#1}}

\knowledgenewrobustcmd{\WP}{\cmdkl{\mathrm{WP}}}
\knowledgenewrobustcmd{\AN}{\cmdkl{\mathrm{AN}}}
\knowledgenewrobustcmd{\ST}{\cmdkl{\mathrm{ST}}}
\knowledgenewrobustcmd{\AT}{\cmdkl{\mathrm{AT}}}
\knowledgenewrobustcmd{\MAX}{\cmdkl{\mathrm{MAX}}}
\knowledgenewrobustcmd{\MIN}{\cmdkl{\mathrm{MIN}}}
\knowledgenewrobustcmd{\THR}[1]{\cmdkl{\mathrm{THR}_{#1}}}
\knowledgenewrobustcmd{\XOR}{\cmdkl{\mathrm{XOR}}}
\knowledgenewrobustcmd{\TheMinion}{\cmdkl{\bra{\mathrm{ARITH}}}}
\knowledgenewrobustcmd{\TheAMinion}{\cmdkl{\dra{\mathrm{ARITH}}}}

\newcommand{\tuple}[1]{\overline{#1}}
\newcommand{\relstr}[1]{{\mathbb{#1}}}
\newcommand{\funset}[1]{\mathcal{#1}}

\newcommand{\funforleq}[2]{\bra{#1}\phantom{}^{#2}_{\le}}

\newcommand{\funfor}[2]{\bra{#1}^{#2}}

\newcommand{\Boolean}{\{0,1\}}

\newrobustcmd\gen[2]{\withkl{\kl[\gen{#1}]}{\cmdkl{\mathrm{#1}^{(#2)}}}}
\knowledge{\gen{wp}}{notion}
\knowledge{\gen{an}}{notion}
\knowledge{\gen{st}}{notion}
\knowledge{\gen{max}}{notion}
\knowledge{\gen{min}}{notion}
\knowledge{\gen{at}}{notion}

\newrobustcmd{\genthr}[2]{\mathrm{thr}_{#1}^{(#2)}}

\knowledgenewrobustcmd{\agenthr}[2]{\cmdkl{\mathrm{thr}_{#1}^{#2}}}
\knowledgenewrobustcmd{\agenat}[1]{\cmdkl{\mathrm{at}_{#1}}}
\newcommand{\agenst}[1]{\tuple{\mathrm{st}}^{#1}}

\newrobustcmd\sless{\prec}
\newrobustcmd\smore{\succ}
\newrobustcmd\less{\preceq}
\newrobustcmd\more{\succeq}
\newrobustcmd\equi{\sim}

\newcommand\+{\mkern8mu}

\theoremstyle{plain}
\newtheorem{theorem}{Theorem}[section]
\newtheorem{lemma}[theorem]{Lemma}
\newtheorem{corollary}[theorem]{Corollary}
\newtheorem{proposition}[theorem]{Proposition}
\theoremstyle{remark}
\newtheorem{claim}[theorem]{Claim}
\newtheorem{observation}[theorem]{Observation}
\crefname{observation}{Observation}{Observations}

\theoremstyle{definition}
\newtheorem*{definition}{Definition}

\definecolor{Empty}{gray}{0.85}
\definecolor{Yellow}{rgb}{0.8,0.6,0.2}
\definecolor{Blue}{rgb}{0.2,0.6,0.8}
\definecolor{Neg}{rgb}{0.85,0.2,0.2}
\definecolor{Pos}{rgb}{0.2,0.8,0.2}
\definecolor{Emph}{rgb}{0.6,0.2,0.4}

\newcolumntype{E}{>{\columncolor{Empty}}c}
\newcolumntype{Y}{>{\columncolor{Yellow!80}}c}
\newcolumntype{B}{>{\columncolor{Blue!80}}c}
\newcolumntype{N}{>{\columncolor{Neg!80}}c}
\newcolumntype{P}{>{\columncolor{Pos!80}}c}

\newcommand{\emc}[1]{\cellcolor{Emph!50}\textcolor{black}{#1}}
\newcommand{\poc}[1]{\cellcolor{Pos!50}\textcolor{black}{#1}}
\newcommand{\nec}[1]{\cellcolor{Neg!50}\textcolor{black}{#1}}

\newcommand{\relation}[2]{      
\begin{tabular}{#1}
    #2
\end{tabular}
}

\let\emptyset\varnothing

\newcommand{\size}[1]{\left|#1\right|}

\newcommand{\twt}[2]{\dra{#1}^{#2}}
\knowledgenewrobustcmd{\ACondition}{\cmdkl{\mathrm{LOW}_{\dra{\cdot}}}}
\knowledge{\ACondition}[separating condition]{notion}
\newcommand{\closure}[1]{\overline{#1}}
\knowledgenewrobustcmd{\AMMAX}{\cmdkl{\dra{\mathrm{WP}}}}
\knowledgenewrobustcmd{\AMMIN}{\cmdkl{\dra{\mathrm{MIN}}}}
\knowledgenewrobustcmd{\AMTHR}[1]{\cmdkl{\dra{\mathrm{THR}}^{#1}}}
\knowledgenewrobustcmd{\AMAT}{\cmdkl{\dra{\mathrm{AT}}}}
\knowledgenewrobustcmd{\MAMAT}{\cmdkl{\dra{-\mathrm{AT}}}}

\DeclareMathOperator{\distOp}{dist}
\knowledgenewrobustcmd{\dist}[2]{\cmdkl{\distOp_{\infty}(#1,#2)}}
\newcommand{\seq}[1]{\tuple a(#1)}
\newcommand{\seqel}[1]{a(#1)}
\newcommand{\tr}[1]{t_{#1}}
\newtheorem*{condition}{Condition}

\DeclareMathOperator{\domOP}{dom}

\knowledgenewrobustcmd{\domainup}[1]{\cmdkl{\domOP{\upharpoonright}(#1)}}
\knowledgenewrobustcmd{\domaindown}[1]{\cmdkl{\domOP{\downharpoonright}(#1)}}

\knowledgenewrobustcmd{\Cminion}{\cmdkl{c-minion}}
\knowledgenewrobustcmd{\Cminions}{\cmdkl{c-minions}}
\knowledge{\Cminion}[\Cminions]{notion}

\knowledgenewrobustcmd{\STl}{\cmdkl{\relstr{ST}_L}}
\knowledgenewrobustcmd{\STr}{\cmdkl{\relstr{ST}_R}}

\AtBeginDocument{%
 \providecommand\BibTeX{{%
   \normalfont B\kern-0.5em{\scshape i\kern-0.25em b}\kern-0.8em\TeX}}}

\begin{document}

\title{%
Injective hardness condition for PCSPs%
}
\titlenote{This  research  was  funded  in  whole or  in  part  by   the National Science Centre, Poland under the  Weave-UNISONO  call  in  the  Weave  programme 2021/03/Y/ST6/00171.  For  the  purpose  of  Open  Access,  the  author  has applied a CC-BY public copyright licence to any Author Accepted Manuscript (AAM) version arising from this submission.}

\author{Demian Banakh}
\email{demian.banakh@doctoral.uj.edu.pl}
\orcid{0009-0008-7159-3735}
\affiliation{%
  \institution{Faculty of Mathematics and Computer Science, Jagiellonian University}
  \streetaddress{Łojasiewicza 6}
  \city{Kraków}
  \country{Poland}
  \postcode{30-348}
}
\affiliation{%
  \institution{Doctoral School of Exact and Natural Sciences, Jagiellonian University}
  \streetaddress{Łojasiewicza 11}
  \city{Kraków}
  \country{Poland}
  \postcode{30-348}
}

\author{Marcin Kozik}
\email{marcin.kozik@uj.edu.pl}
\orcid{0000-0002-1839-4824}
\affiliation{%
  \institution{%
  Faculty of Mathematics and Computer Science, Jagiellonian University}
  \streetaddress{Łojasiewicza 6}
  \city{Kraków}
  \country{Poland}
  \postcode{30-348}
}


\begin{abstract}
We present a template for the Promise Constraint Satisfaction Problem (PCSP) which is NP-hard but does not satisfy the current state-of-the-art hardness condition [ACMTCT'21]. 
We introduce a new ``injective'' condition based on the smooth version of the layered PCP Theorem and use this new condition to confirm that the problem is indeed NP-hard.
    
In the second part of the article, we establish a dichotomy for Boolean PCSPs defined by templates with polymorphisms in the set of linear threshold functions. 
The reasoning relies on the new injective condition.
    
%
\end{abstract}





\maketitle

\section{Introduction}
Many computational questions can be cast as constraint satisfaction problems (CSPs).
In the most general form, an instance of a constraint satisfaction problem consists of variables, values, and constraints restricting admissible evaluations.
The objective is to assign values to variables and satisfy all the constraints.
There are many ways of encoding constraints, 
but the simplest and most common is to list all the allowed tuples of values: 
for example, we can require that among variables $x,y$, and $z$, at least one is True;
in such a case, the constraining relation contains seven tuples: $001,010,100,011,101,110,111$.
Presenting constraints in different ways can alter the problem.
The simplest way to avoid such complications is to restrict allowed constraining relations to a finite set.
This simple definition immediately leads to a problem that had eluded scientists for 20 years: the CSP Dichotomy Conjecture.

The CSP Dichotomy Conjecture, 
proposed in 1998 by Feder and Vardi~\cite{federvardi}, 
stated that, for every relational structure $\relstr B$, 
the CSP template $\relstr B$~%
(i.e., solving CSP instances with relations from $\relstr B$) 
is NP-complete or solvable in a polynomial time.
At that time, two results supported it: 
one theorem due to Schaefer~\cite{schaefer1978} proving the dichotomy for $\relstr B$'s on a Boolean domain, 
and one due to Hell and Ne\v set\v ril~\cite{hellnesetril} for undirected graphs~%
(i.e., $\relstr B$'s with a single, binary, symmetric relation).

The progress on the conjecture was slow until a series~\cite{jeavons1997, jeavons1998, BJK05} of papers established an algebraic approach to CSP.
The idea was to deviate from studying the complexity by combinatorial means 
and to focus on the set of polymorphisms~%
(a multi-variable version of endomorphisms) of template $\relstr B$ instead.
The first series of papers culminated in an article~\cite{BJK05} proposing a universal "hardness condition" for CSP: 
If polymorphisms map homomorphically to projections, then $\CSP(\relstr B)$ is NP-complete 
and otherwise is solvable in a polynomial time.
The authors of the series proved "the easy half" of the conjecture: 
homomorphism to projections implied hardness.
The path forward was clear: 
one needed to take the hardness condition and use its negation to prove tractability.
Nevertheless, it took twenty years of intense development until Bulatov~\cite{bulatov2017dichotomy} and Zhuk~\cite{zhuk2017proof},
in a celebrated work, confirmed "the other half" of the CSP conjecture.

The CSP Dichotomy Theorem of Bulatov and Zhuk tells us that a "slice" of class NP~%
(i.e., the finite domain CSPs) does not contain problems of intermediate complexity and thus is, in fact, somewhat simple.
More than that: it shows that, in this slice, tractability is explained and characterized by elegant algebraic properties~\cite{cyclic, Siggers}.
Is it possible to extend the class and maintain the simplicity? 
Is there a uniform reason for tractability beyond CSPs?
Can we understand a bit larger slice of NP? 
Two main lines of research attempt to accomplish just that.
The first one investigates CSPs defined by carefully selected infinite templates and is beyond the scope of this paper.
The second one attempts to bridge the gap between decision and approximation problems: 
it studies problems known as Promise CSPs (PCSPs).

The PCSPs were introduced in~\cite{austrin20172+sat}.
They significantly extend the class of CSPs 
and contain new and notoriously difficult open problems such as approximate coloring~\cite{garey1976approxcol}.
An example of an approximate coloring problem is: 
find a $6$-coloring of a given $3$-colorable graph~%
(the $3$-coloring is "the promise" and is not a part of the input).
Note that the $6$ and $3$ are parameters; 
every pair of natural numbers defines an approximate coloring problem.

More generally, a Promise CSP depends on a pair of structures: 
we have a target structure $\relstr B$ and a promise structure $\relstr A$.
The computational task changes from "given instance $\relstr I$ find a solution in $\relstr B$" in CSPs 
to "given an instance $\relstr I$ solvable in $\relstr A$ find a solution in $\relstr B$" in PCSPs.
Despite a flurry of publications~\cite{austrin20172+sat,ficak2019symmetric,brakensiek2021conditional, CLAP, Bible, KO}, analogs of the fundamental results that had supported the CSP dichotomy conjecture remain unverified:
\begin{enumerate}
  \item The dichotomy for Boolean PCSPs is an open problem, 
    i.e., we do not have an analog of Shaeffer's result~\cite{schaefer1978}.

  \item The dichotomy for graphs is unknown. 
    The approximate graph coloring problems belong to this class of problems; 
    the computational complexity of approximate graph coloring problems is known only 
    for special classes~\cite{brakensiek2016approxcol, KO},
    or under additional computational assumptions~\cite{dinur2006approxcol,dToOne}.
    The problem in the full generality, 
    i.e., an analog of the result of Hell and Ne\v set\v ril~\cite{hellnesetril}, 
    remains out of reach.

  \item A uniform "hardness condition" for PCSPs remains elusive.%
  \footnote{A recent paper~\cite{newdalmauoprsal} proposes a new condition based on local consistency. We do not know yet how practical or general their proposed condition is.}
    Some criteria were proved sufficient~\cite{austrin20172+sat,Bible, brandts2021llc} and were subsequently applied  
    to show the hardness of PCSPs in many special cases.
    However, we do not know a hardness condition 
    that could confidently define the border of tractability for PCSPs.
\end{enumerate}
In this paper, we contribute to items 1 and 3.
We start discussing our contribution by describing item 3.

\subsection{Hardness conditions}
By the results of Bulatov and Zhuk, a CSP template  $\relstr B$ defines a computational problem that is NP-complete 
if, and only if, (assuming P is different from NP for a second) 
the set of polymorphisms of $\relstr B$ has a minion homomorphism to projections%
\footnote{A minion homomorphism to projections will not be defined here, we will use an equivalent formulation.
}. To state it more precisely:

\AP
A {\em polymorphism}%
\footnotemark 
of $\relstr B$ is a multi-variable homomorphism from $\relstr B$ to $\relstr B$ or, equivalently, a homomorphism from $\relstr B^n$ to $\relstr B$. 
A \intro{choice function}, usually denoted by $I$, is defined on a set of multi-variable functions 
and assigns, to each $f$ in the set, 
a non-empty set $I(f)$ of coordinates of $f$; 
e.g., if  $f:\{0,1\}^7\rightarrow\{0,1\}$ then $I(f)$ can be $\{3,5,7\}$, or $\{1,2,3,4,5,6\}$ etc. 

Let $f$ be a multi-variable function; by identifying, permuting, and introducing dummy variables, 
$f$ can define a new function $g$; 
e.g., $g(x,y,z) = f(x,y,x,y,x,x,x)$. 
The new function $g$ is a {\em minor} of $f$, 
and the description of the process is the minor map $\pi$~
(write $f\xrightarrow{\pi}g$ for ``$g$ is minor of $f$ using map $\pi$''). 
In the running example the domain of $\pi$ is $\{1,\dotsc,7\}$, the codomain is $\{1,2,3\}$, 
and $\pi(1)=\pi(3)=\pi(5)=\pi(6)=\pi(7)=1$, 
$\pi(2)=\pi(4)=2$ so that $g(x_1,x_2,x_3)=f(x_{\pi(1)},\dotsc,x_{\pi(7)})$.

We proceed to define a condition equivalent to having a minion homomorphism to projections. 
The following condition can hold or fail in the set of polymorphisms of $\relstr B$ denoted by $\Pol(\relstr B)$:
\AP\begin{condition}[\intro{Single choice}]
There exists a \kl{choice function} $I$ such that: 
\begin{itemize}
  \item $|I(f)| = 1$ for every $f$ and 
  \item if $f\xrightarrow{\pi}g$ then $\pi(I(f)) = I(g)$. 
\end{itemize}
\end{condition}
\noindent In other words, the \kl{choice function} selects a single coordinate of a polymorphism (the first item) and is compatible with taking minors (the second item): 
Given $f$ and a minor map $\pi$,
selecting a coordinate of $f$ according to $I$ and applying $\pi$, 
or taking the chosen coordinate of the minor produces the same outcome. 
Note that the condition holds in the set of projections of all arities: the \kl{choice function} $I$ selects the dictatorship coordinate.

The CSP dichotomy theorem of Bulatov~ \cite{bulatov2017dichotomy} and Zhuk~\cite{zhuk2017proof} states that tractable CSPs can be characterized by the \kl{single choice condition}:
\begin{theorem}[The CSP dichotomy theorem]
Let $\relstr B$ be a finite relational structure then:
  \begin{itemize}
      \item \kl{single choice} holds in $\Pol(\relstr B)$ and $\CSP(\relstr B)$ is NP-complete, or
      \item \kl{single choice} fails in $\Pol(\relstr B)$ and $\CSP(\relstr B)$ is tractable.
  \end{itemize}
\end{theorem}

Given a PCSP problem defined by a pair of structures $(\relstr A, \relstr B)$, polymorphisms%
\footnotemark[\value{footnote}]
\footnotetext{Definition of a polymorphism can be found in~\cref{sec:notation}.}
are multi-variable homomorphisms from $\relstr A$ to $\relstr B$, 
or equivalently, homomorphisms from $\relstr A^n$ to $\relstr B$. 
Note that, if $\relstr A =\relstr B$, then $\PCSP(\relstr B, \relstr B)$ is $\CSP(\relstr B)$, and the definitions coincide. 

Unfortunately, the \kl{single choice condition} does not characterize NP-hard PCSPs.
Even the first results~\cite{austrin20172+sat, brakensiek2017promise} in the area of PCSPs used a weaker condition.
The weaker condition implied the hardness of the corresponding PCSPs, but until recently~\cite{barto2022babypcp}, 
the hardness proof required the PCP theorem~\cite{arora1998pcp,raz1998rep}.
In contrast, hardness implied by the \kl{single choice} relies on 
a direct reduction from the 3-SAT problem. 
The weaker condition was:
\AP\begin{condition}[\intro{Multiple choice}]
There exists a \kl{choice function} $I$ and a number $M$ such that: 
\begin{itemize}
  \item $|I(f)| \leq M$ for every $f$ and 
  \item if $f\xrightarrow{\pi}g$ then $\pi(I(f))\cap I(g)\neq\emptyset$. 
\end{itemize}
\end{condition}
\noindent In simple terms, instead of one coordinate, $I$ chooses up to $M$ coordinates, 
and taking a minor cannot produce an operation with a disjoint selection. 
The \kl{multi-choice condition} gives several new hardness results e.g.,~\cite{austrin20172+sat,
brakensiek2017promise,
ficak2019symmetric,
KO} 
but, unfortunately, is yet again too restrictive.

A few years later, Brands, Wrochna and \v{Z}ivn\'{y}~\cite{brandts2021llc} improved the hardness condition which allowed them to tackle new classes of PCSPs. 
The new condition uses a layered version of PCP to relax the second item: 
taking a minor once can produce an empty intersection, 
but taking $M$ consecutive minors cannot produce pairwise disjoint sets:
\AP\begin{condition}[\intro{Layered choice}]
There exists a \kl{choice function} $I$ and a number $M$ such that: 
\begin{itemize}
  \item $|I(f)| \leq M$ for every $f$ and 
  \item for every chain of minor maps%
  \footnote{We borrow this notion of \emph{chain of minors} from \cite{brandts2021llc}.}%
  : 
    \begin{equation*}
      f_1\xrightarrow{\pi_{1,2}}f_2\xrightarrow{\pi_{2,3}}\dotsb\xrightarrow{\pi_{M-1,M}}f_M
    \end{equation*}
    there exist $i<j$ such that $\pi_{i,j}(I(f_i))\cap I(f_j)\neq \emptyset$ (where $\pi_{i,j} = \pi_{j-1,j} \circ \dots \circ \pi_{i,i+1}$).
\end{itemize}
\end{condition}
\noindent The new condition was used extensively and proved to be very valuable, e.g.,~\cite{brandts2021llc,ThreeEL}. 

Unfortunately, as we show in this paper, it is not necessary. 
Splitting our claim into parts: we identify a particular PCSP template, 
confirm the failure of the \kl{layered condition}, 
and show that the associated computational problem is NP-hard. 
The template is a pair of Boolean relational structures introduced in~\cref{sec:ex-st}.
To fail the \kl{layered condition}, we provide a structural description of the polymorphisms 
and show that no \kl{choice function} satisfies the \kl{layered condition} in this set. 
Finally, we define a new hardness condition 
and use it to prove that the template is, in fact, NP-hard. 
Our condition relies on the hardness of the smooth version of Label Cover \cite{khot2002smooth}, to which we add layers (see \cref{app:smoothLC}):
\AP\begin{condition}[\intro{Injective layered choice}]
There exists a \kl{choice function} $I$ and a number $M$ such that: 
\begin{itemize}
  \item $|I(f)| \leq M$ for every $f$ and 
  \item for every chain of minor maps:
    \begin{equation*}
      f_1\xrightarrow{\pi_{1,2}}f_2\xrightarrow{\pi_{2,3}}\dotsb\xrightarrow{\pi_{M-1,M}}f_M
    \end{equation*}
    if $\pi_{i,i+1}$ is injective on $I(f_i)\cup\bigcup_{j<i}\pi_{j,i}(I(f_j))$ then
    there exist $i<j$ such that $\pi_{i,j}(I(f_i))\cap I(f_j)\neq \emptyset$.
\end{itemize}
\end{condition}
\noindent The improvement that allows us to apply the new condition where the \kl{layered condition} fails is the injectivity: 
The requirement in the second item can (and will) fail on some chains of minor maps; it suffices that it holds on the injective ones. 
This is a new idea: with the exception of a conditional hardness used in~\cite{brakensiek2021conditional}, the PCSP results rely on compatibility with all minor maps.

In summary, we contribute to the search for an ultimate hardness condition. 
Ideally, a negation of such a condition would provide enough structure~%
(hopefully, just like in the case of CSPs, elegant structure)
for an algorithm to work~%
(if one believes in a dichotomy).
Our new condition 
is supported 
by an example showing the shortcomings of the previous one.
This fact is, perhaps,
more surprising (and equally interesting) than the condition itself.
Nevertheless, the \kl{injective layered condition} condition is good enough to explain hardness for a new class of Boolean templates.

\subsection{Boolean templates}
In contrast to Schaefer's result, the dichotomy for Boolean PCSPs remains an open problem. 
The difference in the complexity of these two problems is quite significant:  
On the one hand, the set $\{\Pol(\relstr B) \mid \relstr B\text{ is a Boolean template}\}$ was fully described by Post in 1941~\cite{Post}.
On the other hand,
understanding a similar set defined by PCSP templates is,
at the moment, 
impossible.
Even with tools provided by the analysis of Boolean functions, one of the strongest results in the area remains partial and conditional~\cite{brakensiek2021conditional}.

Our contribution to this line of research is a classification (resulting in a dichotomy) of Boolean \kl(PCSP){templates} with \kl{minions} consisting of linear threshold functions.
The linear threshold functions are very well behaved in the context of the analysis of Boolean functions --- we investigate them via a lens of PCSP. 
Our result implies, among other things, that the \kl{minions} defining ``the boundary of tractability'' for Boolean symmetric \kl(PCSP){templates} remain ``on the boundary of tractability'' among all Boolean \kl(PCSP){templates} (this sentence, again, assumes $P\neq NP$). 
It is the next step in the process of mapping the landscape of Boolean PCSPs, 
with a Boolean PCSP dichotomy as a long-term goal.
The concepts need to be defined:

A \kl{relational structure} defined over the set $\{0,1\}$ is called a Boolean structure. 
A PCSP \reintro(PCSP){template} $(\relstr A,\relstr B)$ is Boolean if both structures are. 
A Boolean function $f$ is a linear threshold function if: 
\begin{equation*}
  f(\tuple x) = \begin{cases}
    0 & \text{ if } \sum_i a_ix_i  \leq t \\
    1 & \text{ otherwise }.
  \end{cases}
\end{equation*}
for some $\tuple a, t$.
Our first application of the \kl{injective layered condition} is easy to state:
\begin{theorem}
	Let $(\relstr A, \relstr B)$ be a Boolean PCSP \kl(PCSP){template}. 
  If $\Pol(\relstr A,  \relstr B)$ consists of some linear threshold functions, then $\PCSP(\relstr A, \relstr B)$ is NP-hard or solvable in polynomial time.
\end{theorem}
\noindent The second corollary requires some background:

In 2019 a sequence of two papers~\cite{brakensiek2017promise,ficak2019symmetric} established one of the strongest dichotomies in a Boolean case: dichotomy for Boolean symmetric templates~%
(i.e., templates with all relations closed under taking permutations of tuples).
The classification is straightforward: 
a theorem provides a short list~%
(see~\Cref{sec:notation})
of Boolean minions~(i.e., sets of Boolean functions closed under taking minors). Every symmetric $\PCSP(\relstr A,\relstr B)$ with $\Pol(\relstr A,\relstr B)$ above a minion in the list is tractable; all others are NP-hard.

Our new classification implies that if $\Pol(\relstr A,\relstr B)$~%
(for arbitrary, not necessarily symmetric, $\relstr A$ and $\relstr B$) is 
strictly under any of the \kl{minions} in the list then $\PCSP(\relstr A,\relstr B)$ is NP-hard.
Note that this does not establish a Boolean dichotomy: many non-symmetric Boolean templates have polymorphism sets incomparable with minions in the list. 
Nevertheless, the ``vertical'' boundary of tractability remains the same among all Boolean templates and thus the algorithms associated with minions in the list are, in a sense, optimal~%
(unless $P=NP$).

\section{Preliminaries and notation}
\label{sec:notation}
\AP
This section contains basic definitions and notions:
A set $R \subseteq A^n$ over the \emph{universe} $A$ is an $n$-ary relation.
A Cartesian power of such $R$ is a relation $R^m \subseteq (A^m)^n$ such that $(a^1, \dots, a^n) \in R^m$ iff $(a^1_i, \dots, a^n_i) \in R$ for each $i$.
A \intro{relational structure} $\relstr A$ is a tuple $(A; R_1^{\relstr A}, \dots, R_l^{\relstr A})$ where each $R_i$ is a relation over $A$.
A \intro{power} of a relational structure is defined on a power of the universe using powers of corresponding relations.
Two relational structures are \emph{similar} if they have the same sequence of arities of respective relations.
For two similar structures $\relstr A$ and $\relstr B$, a \intro{homomorphism} from $\relstr A$ to $\relstr B$ is a map $h : A \to B$ such that for each $i$ and each tuple $(a_1, \dots, a_k) \in R_i^{\relstr A}$ the tuple $(h(a_1), \dots, h(a_k)) \in R_i^{\relstr B}$.

Next, we define CSPs.
For the sake of brevity, 
we use the homomorphism parlance and not the variable-value one used in the introduction.
There are two versions of a CSP problem:
\begin{definition}[decision CSP]\AP
    $\CSP(\relstr B)$ is the problem of deciding whether an input \kl{structure} $\relstr I$ similar to $\relstr B$ admits a \kl{homomorphism} to $\relstr B$. 
    In this case, $\relstr B$ is called the \emph{template} for $\CSP(\relstr B)$.
\end{definition}

\begin{definition}[search CSP]
  $\CSP(\relstr B)$ is the problem of finding a \kl{homomorphism} from input \kl{structure} $\relstr I$ (similar to $\relstr B$) to $\relstr B$.
\end{definition}
\noindent Similarly, there are two versions of the Promise Constraint Satisfaction Problem:


\begin{definition}[decision PCSP]\AP
    For similar $\relstr A$ and $\relstr B$ with a \kl{homomorphism} $\relstr A \to \relstr B$, 
    the problem $\PCSP(\relstr A, \relstr B)$ is, given an input \kl{structure} $\relstr I$ similar to $\relstr A$ and $\relstr B$, to output YES if $\relstr I \to \relstr A$, and NO if $\relstr I \not\to \relstr B$. In this case, the pair $(\relstr A, \relstr B)$ is called the \intro(PCSP){template} for $\PCSP(\relstr A, \relstr B)$.
\end{definition}

\begin{definition}[search PCSP]
    For similar $\relstr A$ and $\relstr B$ with a \kl{homomorphism} $\relstr A \to \relstr B$, 
    the problem $\PCSP(\relstr A, \relstr B)$ is, given an input \kl{structure} $\relstr I$ that has a \kl{homomorphism} to $\relstr A$, 
    output a \kl{homomorphism} from $\relstr I$ to $\relstr B$.
\end{definition}
\noindent The search problem is at least as hard as the decision problem.
We will be providing algorithms for the search version and proofs of hardness for the decision version.

\AP
A function $f$ is a \intro{polymorphism} of, or \reintro{compatible with}, $(\relstr A, \relstr B)$ if it is a \kl{homomorphism} from a \kl{power} of $\relstr A$ to $\relstr B$.
The set of all such polymorphisms is denoted by $\Pol(\relstr A, \relstr B)$.
Note that $\Pol(\relstr A, \relstr B)$ effectively captures the complexity of $\PCSP(\relstr A, \relstr B)$:
\begin{theorem}[See Theorem 6.2 in \cite{brakensiek2017promise}]
    If $\Pol(\relstr A, \relstr B) \subseteq \Pol(\relstr A', \relstr B')$, then $\PCSP(\relstr A', \relstr B')$ is poly-time reducible to $\PCSP(\relstr A, \relstr B)$.
\end{theorem}
\noindent\AP For any $\relstr A$ and $\relstr B$, the set $\Pol(\relstr A, \relstr B)$ is a \intro{minion}, i.e. a function set closed under taking \emph{minors}:
\begin{definition}[Minor, minor map]\AP
    An $n$-ary function $g : A^m \to B$ is a \intro{minor} of $n$-ary $f : A^n \to B$ given by a \intro{minor map} $\pi : [n] \to [m]$ if
    \begin{equation*}
        g(a_1, \dots, a_m) = f(a_{\pi(1)}, \dots, a_{\pi(n)}), 
        \text{ for all  } a_1, \dots, a_m \in A.
    \end{equation*}
    In this case we say that $g$ is a $\pi$-minor of $f$, or $f \xrightarrow{\pi} g$.
\end{definition}
\noindent\AP A \intro{projection}, or \reintro{dictator}, is a function $f:A^n\rightarrow A$ such that, for some $i$, $f(a_1,\dotsc,a_n) = a_i$ for all $a_1,\dotsc,a_n$.
A function $f:A^n\rightarrow A$ is \intro{idempotent} if $f(a,\dotsc,a)=a$ for all $a\in A$.
Every minor of a projection is a projection, and every minor of an idempotent function is idempotent. 

\AP
Finally, we say that a coordinate $i$ of a function $f : A^n \to B$ is \intro{essential} if there exist $a_1, \dots, a_n, b_i \in A$ such that $f(a_1, \dots, a_i, \dots, a_n) \neq f(a_1, \dots, b_i, \dots a_n)$. A function $f$ doesn't depend on its non-essential coordinates.

\subsection{Boolean notation}
Recall that a \kl{relational structure} $\relstr A$ is Boolean if $A = \{0,1\}$ and a function is Boolean if its domain and codomain are both $\{0,1\}$.
In this paper we mostly work with Boolean structures, functions and \kl{minions}. 
We say that $n$-ary Boolean function $f$ is
\begin{itemize}
    \itemAP \intro{folded} if $f(a_1, \dots, a_n) = 1 - f(1 - a_1, \dots, 1 - a_n)$ for all $a_1, \dots, a_n \in A$,
    \itemAP \intro{monotone} \emph{in coordinate $i$} if
      \begin{multline*}
        f(a_1,\dotsc,a_{i-1},0,a_{i+1},\dotsc,a_n) \leq
        f(a_1,\dotsc,a_{i-1},1,a_{i+1},\dotsc,a_n),
      \end{multline*}
    \item \reintro{antimonotone} \emph{in coordinate $i$} if function $1-f$ is monotone in~$i$.
\end{itemize}
\AP
To simplify the presentation we will be using the bra-ket notation: $\braket{ \tuple a | \tuple x} = \sum_i a_ix_i$.
With this at hand, for every tuple $\tuple a$ and a threshold $t$ we define a Boolean function:
\begin{equation*}
  \funforleq{\tuple a}{t}(\tuple x) = \begin{cases}
    0 & \text{ if } \braket{\tuple a | \tuple x } \leq t \\
    1 & \text{ otherwise}.
  \end{cases}
\end{equation*}
\noindent We let $\intro*\TheMinion$ denote the \kl{minion} of all Boolean functions defined in this way.

In the previous work, the following Boolean \kl{minions} were of prominence:
\AP\begin{itemize}
  \item The minion of all constant $0$ functions, and the minion of all constant $1$ functions,
  \item the minion $\intro*\MAX$ generated by the set of all max functions,
  \item the minion $\intro*\MIN$ generated by the set of all min functions,
  \item the minion $\intro*\XOR$ generated by the set of all xor functions,
  \item the minion $\intro*\AT$ generated by the set of all alternating threshold functions, i.e., functions $\funforleq{1,-1,1,-1,\dotsc,1}{0}$ of odd arities,
  \item minions $\intro*\THR{q}$ parameterized by a rational number $0<q<1$, each containing all threshold $q$ functions, i.e., $\funforleq{1,\dotsc,1}{qn}$ where $n$, the arity, is co-prime with the denominator of $q$.
\end{itemize}
The list is completed by the \intro{negated} versions of all the minions;
the negation of $f$ is $\tuple x\mapsto 1-f(\tuple x)$ and the negation of a minion consists of negations of all its members.   

\section{The main results}
The general hardness result is easy to state, and the proof (in \Cref{app:smoothLC}) follows from \textit{smooth} LC \cite{khot2002smooth} (after adding layers):
\begin{theorem}
  Let $(\relstr A,\relstr B)$ be a PCSP \kl(PCSP){template}. 
  If the \kl{injective layered condition} holds in $\Pol(\relstr A,\relstr B)$ then $\PCSP(\relstr A,\relstr B)$~%
  (in the decision version)
  is NP-hard.
\end{theorem}
\noindent
In~\cref{sec:ex-st} we present an example distinguishing the new condition from the \kl{layered condition}.
To show that the \kl{layered condition} fails, we investigate the \kl{minion} of \kl{polymorphisms}. 
The most technical part of this reasoning is postponed until~\Cref{app:ST}.
Hardness follows from our new Boolean classification which is introduced in the next paragraph.

The condition for tractability for \kl{minions} $\minion M\subseteq\TheMinion$ is:
\begin{equation}
  \tag*{$\intro*\PCondition$}
  \label{eq:PCondition}
  \forall{\varepsilon > 0}~\exists~\tuple a, t\ \text{ s. t. }\ \funforleq{\tuple a}{t}\in\minion M \text{ and } \max_i\size{a_i} \le \varepsilon \sum_i\size{a_i}
\end{equation}
and the classification is:
\begin{theorem}\label{thm:main}
  Let $(\relstr A, \relstr B)$ be Boolean \kl(PCSP){template} such that $\Pol(\relstr A,\relstr B)\subseteq\TheMinion$, then:
  \begin{itemize}
    \item $\PCondition$ holds for $\Pol(\relstr A,\relstr B)$ and $\PCSP(\relstr A, \relstr B)$ is solvable in polynomial time, or
    \item $\PCondition$ fails in $\Pol(\relstr A,\relstr B)$ and $\PCSP(\relstr A, \relstr B)$ is NP-hard.
  \end{itemize}
  In the first case constant, $\MIN$, $\MAX$, $\AT$, or $\THR{t}$ (for some real number $t$), 
  or their \kl{negated} counterparts, are included in $\Pol(\relstr A,\relstr B)$.
\end{theorem}
\noindent A proof of this theorem is in~\Cref{sec:proof}; some technical parts of the hardness reasoning are deferred to~\cref{app:hard},
and one additional argument on tractability is in~\Cref{app:Algorithms}.
The claimed consequence is immediate:
\begin{corollary}
  Let $\Pol(\relstr A, \relstr B)$ be properly included in one of the \kl{minions} $\MIN$, $\MAX$, $\XOR$, $\AT$, $\THR{t}$ or their \kl{negated} counterparts.
  Then $\PCSP(\relstr A,\relstr B)$ is NP-hard.
\end{corollary}
\begin{proof}
  For $\MIN, \MAX$, and $\XOR$ the conclusion is immediate: any proper subminion contains operations that depend only on a bounded number of coordinates
  and therefore is NP-hard by \kl{multiple choice condition}.
  Both $\AT$ and $\THR{t}$ are in the scope of~\Cref{thm:main}, and the conclusion follows.
\end{proof}

We conclude this section with two examples.
Both serve as an illustration of the problems we encounter while working with $\TheMinion$.
The first example and its analysis is, in large part, folklore; 
we present it for completeness.
The second example proves that the \kl{injective layered condition} is strictly weaker than the injective condition%
\footnote{We didn't define the injective condition, but it follows from the \kl{injective layered condition} by restricting the second item to chains of length 2 (i.e. a single minor)}.

\subsection{Example 1: layers and presentations}
In this section, we work with a \kl{minion} generated by the functions
\begin{equation*}
  \funforleq{\frac{1}{3},\underbrace{\frac{1}{3m},\dotsc,\frac{1}{3m}}_{2m}}{0.5}
\end{equation*}
and denote it by a provisional name $\intro*\WP$.

Note that in each of the generating operations, the first coordinate has a large impact on the result.
On the other hand, after identifying more than $\frac{3}{2}m$ of the remaining coordinates we obtain a coordinate 
with the coefficient $>0.5$ which makes the \kl{minor} a \kl{projection}, 
and the formerly ``heavy'' coordinate does not impact the result of this projection at all.
A minor extension of this reasoning shows that hardness cannot be derived here from the \kl{multi-choice condition}:
\begin{observation}
  Let $I$ be a \kl{choice function} for $\WP$.
  If there exists $M$ such that $\size{I(f)} \leq M$ for all $f\in\WP$,
  then there exists $f\in\WP$ and a \kl{minor map} $\pi$ such that $\pi(I(f))\cap I(\pi(f))=\emptyset$ (where $\pi(f)$ is a $\pi$-\kl{minor} of $f$).
\end{observation}
\begin{proof}
  Let $I$ be a \kl{choice function} for $\WP$ with the sizes of images bound by $M$. 
  Let $m\geq2M$, let $f$ be an $(2m+1)$-ary generator of $\WP$, let $g$ be an $(M+2)$-ary \kl{projection} on the first coordinate 
  and let $1\neq j \notin I(g)$.
  We define a \kl{minor map} $\pi: [2m+1] \rightarrow [M+2]$ by putting $\pi(i) = j$ if $i\in I(f)$ and $\pi(i) = 1$ for $i\notin I(f)$.
  Note that $f\xrightarrow{\pi}g$ as the sum of weights of coordinates in $I(f)$ accounts for less than half of the total weight.
  By the construction $\pi(I(f))=\{j\}$, and as $j\notin I(g)$ the proof is concluded.
\end{proof}

To show that the \kl{minion} $\WP$ is hard we use~\cref{thm:main}, and to do so we need to confirm that $\PCondition$ fails in $\WP$.
Note that we need to deal with every presentation 
of elements of $\WP$, and a very heavy coordinate in one presentation might be not that heavy in another:
E.g. for suitably small $\varepsilon$:
\begin{equation*}
  \funforleq{\frac{1}{3},\frac{1}{6}, \frac{1}{6}, \frac{1}{6}, \frac{1}{6}}{0.5}\  = \ 
  \funforleq{\frac{1}{5} + 4\varepsilon, \frac{1}{5} - \varepsilon, \frac{1}{5} - \varepsilon, \frac{1}{5} - \varepsilon, \frac{1}{5} - \varepsilon}{0.6}
\end{equation*}

Before we are ready to show $\PCondition$ fails in $\WP$, we need an intermediate step.
It concerns symmetric functions: a function is symmetric if permuting its arguments does not change the result of the application.
\begin{observation}
  No symmetric $5$-ary function belongs to $\WP$. 
\end{observation}
\begin{proof}
  Let $\funforleq{\tuple a}{0.5}\in\WP$ of be of arity $5$ and let $\tuple a$ be a presentation obtained from one of the generating tuples i.e.:
  $\sum_i a_i = 1$ and $\max_i a_i \geq 1/3$.
  Assume wlog that $a_1\geq \frac{1}{3}$ and note that if $a_i>\frac{1}{6}$ for $i\geq 2$ then $a_i>\frac{1}{6}$ for all $i$,
  as the first condition implies that on all permutations of $(1,1,0,0,0)$ the function produces $1$.
  But then $\sum_i a_i> 1/3 + 4/6 = 1$ which is a contradiction.

  Thus $a_i\leq \frac{1}{6}$ for $i=2,..,5$ but then $\braket{\tuple a| (0,0,1,1,1)} \leq .5$~%
  (i.e., the function outputs $0$ on $(0,0,1,1,1)$ and its permutations) and thus 
  both $a_1+a_2+a_3\leq 0.5$  and $a_1+a_4+a_5\leq 0.5$.
  This implies that $\sum_i a_i < 2a_1 + \sum_{i\geq 2} a_i \leq 1$ which, again, is a contradiction.
\end{proof}
\noindent We are ready to tackle the failure of $\PCondition$.
\begin{observation}
  If $\funforleq{\tuple a}{t}\in\WP$ then $\max_i |a_i| > \frac {1}{80}\sum_i\size{a_i}$.
\end{observation}
\begin{proof}
  Let $\funforleq{\tuple a}{t} \in \WP$.
  For the sake of contradiction, assume that all $|a_i| \le \frac{1}{80} \sum_i |a_i|$. We are going to prove that $\funforleq{\tuple a}{t}$ has a 5-ary symmetric \kl{minor}.

  Fix the presentation $\tuple b \ge 0$ obtained from one of the generating tuples such that $\funforleq{\tuple a}{t} = \funforleq{\tuple b}{0.5}$ and $\sum b_i = 1$. 
  By permuting coordinates of both functions, we can assume that $b_1 \ge \dots \ge b_n$. Note that $\sum_I a_i > t \iff \sum_I b_i > 0.5$ for any $I \subseteq [n]$.
  
  Let $\tuple a'$ be obtained by replacing each negative entry in $\tuple a$ with 0. Note that $\funforleq{\tuple a}{t} = \funforleq{\tuple a'}{t}$ since $\WP$ consists only of \kl{monotone} functions. We have $0 \le a'_i \le \frac{1}{80} \sum_i |a_i| < \frac{1}{40} \sum_i a'_i$ for all $i$, where the last inequality follows from $\sum_i a_i > t \ge 0$ (\kl{idempotency}). Denote $\delta := \frac{1}{40} \sum_i a'_i$. We now argue that
  \begin{equation*}
      \sum_i a'_i - \delta < 2t < \sum_i a'_i + 2\delta
  \end{equation*}
  Let $j$ be an index such that $\sum_{i=1}^{j-1} b_i \le 0.5 < \sum_{i=1}^j b_i$.  
  We argue that $\sum_{i=1}^{j-1} a'_i + \sum_{i=j+1}^n a'_i \le 2t$ and $\sum_{i=1}^j a'_i + \sum_{i=j-1}^n a'_i > 2t$; 
  indeed, each equality follows from the fact that $\funforleq{\tuple a'}{t} =\funforleq{\tuple b}{0.5}$ and the index $j$ splits $b_i$'s in half.
  This finishes the proof of the double inequality above.

  We define a \kl{minor} $\funforleq{c_1, \dots, c_5}{t}$ by partitioning the coordinates of $\tuple a'$ in 5 more-or-less equal groups: the total weight of each of them lies in $\frac{1}{5} \sum_i a'_i \pm \delta$. It suffices to verify that $\funforleq{c_1, \dots, c_5}{t}$ is a 5-ary majority function by using the aforementioned bounds on $t$: in particular, on all permutations of $(1, 1, 0, 0, 0)$ the function produces 0, and on all permutations of $(1, 1, 1, 0, 0)$ it produces 1.
\end{proof}
\cref{thm:main} implies that the \kl{minion} $\WP$ is hard, and
the first observation proves that the \kl{multi-choice condition} fails in $\WP$.
We do not need the full power of~\cref{thm:main}; the hardness of $\WP$ can be derived from the \kl{layered choice condition}.
The following example shows that the \kl{layered choice condition} is too strong to capture all hard PCSPs.

\subsection{Example 2: layered hardness is not enough}\label{sec:ex-st}
We work with a \kl{minion} generated by the functions
\begin{equation*}
    \funforleq{1,\underbrace{2,-2,4,-4,\dotsc,2^m,-2^m}_{2m}}{0}
\end{equation*}
and denote it by $\intro*\ST$. \sloppy
This \kl{minion} was identified by Szymon Stankiewicz.

Note that if one identifies, in a \kl{minor}, coordinates with coefficients $k$ and $-k$ the new coordinate will have coefficient $0$, and no function depends on such a coordinate.
Moreover, it is not hard to see that if one identifies two such coordinates in one of the generators of $\ST$~%
(and then identifies the new \kl{non-essential} coordinate with any other coordinate)
one obtains the generator of arity smaller by 2.
The special coordinate with coefficient $1$ can be identified with the coordinate with coefficient $-2$;
the resulting \kl{minor} does not depend on the new identified pair.
As before we can identify the \kl{non-essential} coordinate with anything and obtain the previous generator.
A minor extension of these observations can be used to show that hardness by the \kl{layered choice condition} fails in $\ST$:
\begin{observation}
  Let $I$ be a \kl{choice function} for $\ST$.
  If there exists $M$ such that $\size{I(f)} \leq M$ for all $f\in\ST$,
  then there is a chain of \kl{minor maps} $f_1\xrightarrow{\pi_{1,2}}f_2\xrightarrow{\pi_{2,3}}\dotsb\xrightarrow{\pi_{M-1,M}} f_M$ in $\ST$ such that,
  putting $\pi_{i,j} = \pi_{j-1,j}\circ\dotsb\circ\pi_{i,i+1}$,
  we have 
  \begin{equation*}
    \pi_{i,j}(I(f_i))\cap I(f_j)=\emptyset \text{ for every } i < j
  \end{equation*}
\end{observation}
\begin{proof}
  Let $I$ be a \kl{choice function} for $\ST$ with the sizes of images bound by $M$. 
  Let $m_1 > \dots > m_L > M$ be large enough, let $f_i$ be an $m_i$-ary generator of $\ST$,
  and let $1 \neq l_i \notin I(f_i)$.
  We will define \kl{minor maps} $\pi_{i, i+1}: [m_{i}] \rightarrow [m_{i+1}]$ so that $\pi_{i,i+1}(I(f_i) \cup l_i) = \{l_{i+1}\}$.
  Then $\pi_{i,j}(I(f_i)) = \pi_{j-1,j} \circ \dots \circ \pi_{i,i+1}(I(f_i)) = \{l_j\}$, and as $l_j \notin I(f_j)$ the proof would be concluded.

  To define $\pi_{i,i+1}$, we first identify the coordinate of $f_i$ with coefficient 1 with all coordinates in $I(f_i) \cup l_i$ and their opposites. As discussed, such a \kl{minor} is a generator of a smaller arity, in which the image of $I(f_i) \cup l_i$ corresponds to the coordinate with coefficient 1. Next we identify it with the coordinate with coefficient $-2$ into $l_{i+1}$, leaving the remaining coordinates untouched. The resulting \kl{minor} is again a generator of a smaller arity i.e. $f_{i+1}$.
\end{proof}

\AP
We need one more notation before we can state the main theorem characterizing $\ST$. 
For a Boolean function $f$ we call $(i,j)$ a \intro{fixing pair} if $f(\tuple a) = 0$ whenever $a_i=0$ and $a_j=1$,
and $f(\tuple a) = 1$ whenever $a_i=1$ and $a_j=0$.
Note that for a generator of $\ST$ the last two coordinates form such a pair.
Let $\reintro*\STl, \reintro*\STr$ be the following Boolean \kl{relational structures}:
  \begin{align*}
      \reintro*\STl &= (R\,\hphantom{AE_6}, \kinl{2}{4}\+\,) \\
      \reintro*\STr &= (\NAE 6, \NAE 4)
  \end{align*}
  where
  \begin{equation*}
      R =
      \relation{| Y Y | B B B |}{
          \hline
          1 & 0 & 1 & 0 & 0\\
          0 & 1 & 1 & 0 & 0\\
          1 & 0 & 0 & 1 & 0\\
          0 & 1 & 0 & 1 & 0\\
          1 & 0 & 0 & 0 & 1\\
          0 & 1 & 0 & 0 & 1\\
          \hline
      } \hspace{4em}
      \kinl{2}{4} = 
      \relation{| B B B B B B |}{
          \hline
          1 & 0 & 1 & 0 & 1 & 0\\
          1 & 0 & 0 & 1 & 0 & 1\\
          0 & 1 & 1 & 0 & 0 & 1\\
          0 & 1 & 0 & 1 & 1 & 0\\
          \hline
      }
  \end{equation*}
  \AP
  (every column corresponds to a tuple in the relation) and $\intro*\NAE{k}$ consists of all non-constant $k$-tuples. 
  We present a theorem that connects $\ST$ with the \kl{relational structures} $\STl, \STr$ and a condition using \kl{fixing pairs}.
\begin{restatable}{theorem}{thmST}
    \label{thm:ST}
    For every Boolean function $f$ the following are equivalent:
    \begin{enumerate}
      \item $f\in\ST$,
      \item $f\in\Pol(\STl, \STr) \cap \Pol(\{0\}) \cap \Pol(\{1\})$,
      \item $f$ is \kl{idempotent} and every non-unary \kl{minor} of $f$ has a \kl{fixing pair}~%
      (note that $f$ is its own \kl{minor} as well).
    \end{enumerate}
\end{restatable}
\noindent \cref{app:ST} contains a proof of this theorem.
In this section, we use it to prove the hardness of $\ST$:
\begin{observation}
  Let $(i,j)$ be a \kl{fixing pair} for $\funforleq{\tuple a}{t}$ then
  \begin{equation*}
  \max(|a_i|,|a_j|) > \frac{1}{4}\sum_k|a_k|.
  \end{equation*}
\end{observation}
\begin{proof}
  By definition of a \kl{fixing pair}, we have
  \begin{equation*}
      \begin{cases}
          a_j + \sum\limits_{k \notin \{i, j\} : a_k > 0} a_k &\le t \\
          a_i + \sum\limits_{k \notin \{i, j\} : a_k < 0} a_k &> t
      \end{cases}
  \end{equation*}
  After subtracting these 2 inequalities we get $a_i - a_j > \sum_{k \notin \{i, j\}} |a_k| \ge 0$, and
  \begin{equation*}
      |a_i| + |a_j| \ge |a_i - a_j| = a_i - a_j > \sum_k |a_k| - |a_i| - |a_j|
  \end{equation*}
  Hence $|a_i| + |a_j| > \frac{1}{2} \sum_k |a_k|$,
  so $\max(|a_i|, |a_j|) > \frac{1}{4} \sum_k |a_k|$.
\end{proof}
To sum up: $\ST$ is a \kl{minion} of \kl{polymorphisms} of a finite \kl{relational structure},
defines a $\PCSP$ which is NP-hard by~\cref{thm:main}, and the hardness cannot be derived via the \kl{layered choice condition}.
Note that the hardness proof could be streamlined: 
We don't need the full power of~\cref{thm:main};
a smooth, non-layered version of PCP suffices.

\section{Tuple minions (towards~\crtCref{thm:main})}
\label{sec:Abstract}
In this section, we start working towards a proof of~\Cref{thm:main}.
The first part of the reasoning disregards the fact, that many presentations can define the same function~%
(i.e., it might happen that $\funforleq{\tuple a}{t} = \funforleq{\tuple b}{s}$ for $\tuple a\neq \tuple b$ and $s\neq t$).
We will work just with the tuples.

\AP
For a tuple $\tuple a$ of length $n$ and a number $t$, 
we use a special notation $\twt{\tuple a}{t}$ to denote $(n+1)$-tuple $(a_1,\dotsc,a_n,t)$.
The set of all such tuples is $\intro*\TheAMinion$.
For $\twt{\tuple a}{t}$ and $\twt{\tuple b}{t}$ we write $\twt{\tuple a}{t}\xrightarrow{\pi}\twt{\tuple b}{t}$
\phantomintro(c){minor}
if $\pi:[n]\rightarrow [m]$ and $b_i=\sum_{j:\pi(j) = i} a_j$.
A set $\minion M \subseteq \TheAMinion$ is a \intro*\Cminion\ if for every $\twt{\tuple a}{t}\in\minion M$:
\begin{itemize}
  \item if $\twt{\tuple a}{t}\xrightarrow{\pi}\twt{\tuple b}{t}$ for some $\pi$, 
    then $\twt{\tuple b}{t}\in\minion M$
  \item $\twt{c\tuple a}{ct}\in\minion M$ for every $c>0$~%
    (if $\tuple a = (a_1,\dotsc,a_n)$ then $c\tuple a = (ca_1,\dotsc,ca_n)$).
\end{itemize}
The \Cminion\ \emph{generated by} a set of $\twt{\tuple a}{t}$ is the smallest (under inclusion) 
\Cminion\ that contains it.
By $-\minion M$ we denote $\{\twt{-\tuple a}{-t} : \twt{\tuple a}{t} \in \minion M\}$; note that $\funforleq{-\tuple a}{-t} = 1 - \funforleq{\tuple a}{t}$.
We introduce an analog of $\PCondition$ for \Cminions\:
\begin{equation}
  \label{eq:ACondition}
  \tag*{$\intro*\ACondition$}
  \forall{\varepsilon > 0}~\exists~\tuple a, t\ \text{ s. t. }\ \twt{\tuple a}{t}\in\minion M \text{ and } \max_i|a_i| \leq \varepsilon \sum_i |a_i|
\end{equation}
and start to build the library of tools needed for~\cref{thm:main}:
\subsection{The ``tractable'' case}
We use $l_{\infty}$ distance; if $\twt{\tuple a}{t}$ and $\twt{\tuple b}{s}$ are of the same  arity $n$, then
\AP\begin{equation*}
  \intro*\dist{\twt{\tuple a}{t}}{\twt{\tuple b}{s}} = \max\{|a_1-b_1|,\dotsc,|a_n-b_n|, |t- s|\}
\end{equation*}
The topological closure of a \Cminion\ $\minion M$~%
(in the topology defined by the distance)
is $\minion M$ together with the all limit points; it is denoted by $\closure{\minion M}$.
Next we introduce ``constants'' and analogs of $\AT$ and $\THR{t}$:
\begin{itemize}
  \itemAP any $\twt{\tuple a}{t}$ with $t\leq\min(0,\sum_i a_i)$ or $t\geq\max(0,\sum_ia_i)$ is called \intro(c){constant};
  \itemAP $\twt{\frac{1}{m},\dotsc,\frac{1}{m}}{t}$~(of arity $m$) is $\intro*\agenthr{m}{t}$ and the 
  \Cminion\
  generated by them is $\intro*\AMTHR{t}$;
  \itemAP $\twt{\frac{1}{m},\dotsc,\frac{1}{m},-\frac{1}{m-1},\dotsc,-\frac{1}{m-1}}{0}$~(of arity $2m-1$) is $\intro*\agenat{m}$ and the \Cminion\ generated by them is $\intro*\AMAT$.
\end{itemize}

\noindent The following observation~%
(which we leave without a proof), 
facilitates the subsequent theorem:
\begin{observation}\label{obs:grouping}
  Let $\tuple a$ be a tuple of length $n$ such that $|1-\sum_i a_i|\leq \varepsilon$ and $0\leq a_i\leq\varepsilon$ for every $i$.
  For every $m$ there exists $\pi:[n]\rightarrow [m]$ such that $|\frac{1}{m} - \sum_{i:\pi(i) = j} a_i|\leq \varepsilon$ for every $j$.
\end{observation}
\noindent The theorem is:
\begin{theorem}\label{thm:amainEZ}
  If $\ACondition$ holds  in \Cminion\ $\minion M$, then
  $\minion M$ contains a \kl(c){constant}, or the topological closure $\closure{\minion M}$ includes $\AMTHR{t}$ for some $t$, or $-\AMTHR{t}$ for some $t$, or $-\AMAT\cup\AMAT$.
\end{theorem}
\begin{proof}
  Let $\seq{i},\tr{i}$ be provided by $\ACondition$ for $\varepsilon = \frac{1}{i}$.
  In the first step we substitute each pair $\seq{i},\tr{i}$~%
  (both the elements of the tuple and the threshold), scaling them by $\frac{2}{\sum_j|\seqel{i}_j|}$~%
  (if $\sum_j|\seqel{i}_j|= 0$ we have a \kl(c){constant} and the proof is finished).
  Now, for every $i$, we have $\sum_j|\seqel{i}_j| = 2$.

  If one of the new $\tr{i}$'s is $\geq 2$ or $\leq -2$ we have a \kl(c){constant} and the proof is concluded.
  Otherwise $\{\tr{i}\}_i$ is a bounded sequence, 
  and we substitute $\{\twt{\seq{i}}{\tr{i}}\}_i$ with a subsequence chosen so that the sequence of $\tr{i}$'s converges, say to $T$.
  Similarly the sequence $\{\sum_j\seqel{i}_j\}_i$ is bounded and we can refine $\{\twt{\seq{i}}{\tr{i}}\}_i$ once more so that it converges as well, say to $S$.
  The proof splits into two parts depending on whether $S=0$.

  We consider the case of $S\neq 0$ first; note that $S\geq T \geq 0$ or $S\leq T \leq 0$~%
  (as otherwise we would have a \kl(c){constant} in $\minion M$).
  Both cases are identical and we consider the first one only i.e. $S\geq T \geq 0$.
  We will show that $\AMTHR{T/S}\subseteq\closure{\minion M}$.

  Fix $m$ and $\varepsilon$, and our goal is to construct $\twt{\tuple a}{t}\in\minion M$ such that 
  $\dist{\twt{\tuple a}{t}}{\agenthr{m}{T/S}}\leq \varepsilon$.
  To that end we choose $i$ such that $|\tr{i}-T|\leq S\varepsilon$,
  $|\seqel{i}_j|\leq S\varepsilon$ for every $j$
  and $|\sum_j \seqel{i}_j - S|\leq S\varepsilon$.
  Next we define a \kl{minor map} $\pi$ in two steps:
  In the first step, 
  we merge all the negative coordinates of $\seq{i}$ and 
  then merge them with positive coordinates so that the new coordinate~%
  (which is a sum) is between
  $0$ and $S\varepsilon$.
  In the second step, we use~\cref{obs:grouping} to define further merges and obtain a \kl(c){minor} $\pi(\seq{i})$ with every coordinate at most $S\varepsilon$ distant from $\frac{S}{m}$.
  The $\twt{\tuple a}{t}$ is obtained by scaling this $\pi$-\kl(c){minor} of $\seq{i}$ by $1/S$.
  This finishes the case $S\neq 0$.

  If $S=0$ then $T=0$, as otherwise there is a \kl(c){constant} in $\minion M$ and we are done.
  Fix $m$ and $\varepsilon$ and our goal is to construct $\twt{\tuple a}{t}\in\minion M$ such that 
  $\dist{\twt{\tuple a}{t}}{\agenat{m}}\leq \varepsilon$~%
  (the case of $-\agenat{m}$ is identical and we leave it to the reader).
  To that end we choose $i$ such that $|\tr{i}|\leq \varepsilon, 
  |\seqel{i}_j|\leq \varepsilon$ for every $j$ and $|\sum_{j:\seqel{i}_j>0}\seqel{i}_j- 1|\leq\varepsilon$
  and $|\sum_{j:\seqel{i}_j<0}\seqel{i}_j+1|\leq\varepsilon$.
  We use~\cref{obs:grouping} to define $\pi$ which groups positive coordinates into $m$ groups of weight, approximately, $\frac{1}{m}$,
  and then use it again to, independently, group negative coordinates into $(m-1)$ groups of weight roughly $\frac{-1}{m-1}$.
  The resulting \kl(c){minor} of $\seq{i}$ can be, clearly, taken for $\twt{\tuple a}{t}$ and the proof is finished.
\end{proof}

\subsection{The ``hard'' case}
  In this section, we present a slightly more technical tool.
  It will be used to prove hardness in~\cref{sec:proof}.
  We introduce two new concepts:
  Let $\twt{\tuple a}{t}\xrightarrow{\pi}\twt{\tuple b}{t}$ with $\tuple a$ of length $n$ and let $I\subseteq [n]$;
  \begin{itemize}
    \itemAP if $\size{\pi(I)} = \size{I}$ then $\pi$ is \intro{injective on} $I$~%
    (denoted $\twt{\tuple a}{t}\xrightarrow[I]{\pi}\twt{\tuple b}{t}$),
    \itemAP if for all $j\notin \pi(I)$ we have $\size{\pi^{-1}(j)} = 1$ then $\pi$ is \intro{covered by} $I$. 
  \end{itemize}
  Introducing injective minors is motivated by the \kl{minion} $\ST$.
  Note that a \kl{fixing pair} of coordinates for a function cannot vanish%
  \footnote{More precisely: no function can have two, totally disjoint \kl{fixing pairs}, and under a \kl{minor} \kl{injective on} a \kl{fixing pair}, the image of the pair is \kl{fixing}.}
  after taking a \kl{minor} that is injective on the two selected \kl{fixing} coordinates. 
  On the other hand, such important coordinates can vanish under general \kl{minors}.

  \AP
  Moreover for any $\twt{\tuple a}{t}$~%
  ($\tuple a$ of length $n$)
  we call $I\subseteq [n]$ a \intro{dominating set} 
  if for every $i\in I$ and $j\notin I$ we have $|a_i|\geq |a_j|$.
  A coordinate $i$ is \reintro{dominating} if $\set{i}$ is.

\begin{theorem}\label{thm:amainHARD}
  Let $\ACondition$ fail in $\minion M\subseteq\TheAMinion$ with $\varepsilon >0$.
  If $\minion M\ni \twt{\tuple a}{t}\xrightarrow{\pi}\twt{\tuple b}{t}$ and $I$ are such that:
  \begin{enumerate}
    \item $\pi$ is \kl{injective on} $I$, and 
    \item $\pi$ is \kl{covered by} $I$, and
    \item $I$ \kl{dominates} $\twt{\tuple a}{t}$, and 
    \item $\size{I}\geq \frac{2}{\varepsilon}$
  \end{enumerate}
  then an element of $\pi(I)$ \kl{dominates} $\twt{\tuple b}{t}$.
\end{theorem}

\begin{proof}
  Suppose, for a contradiction, that the conclusion fails and let $w = \max_{i:i\notin I} |a_i| \ge \max_j |b_j|$.
  Our assumption implies that, for every $i\in I$, $|\sum_{j:\pi(j)=\pi(i)}a_j| < w \le |a_i|$.

  Now we define a new \kl{minor map} $\rho$ which is best described as a partition of $[n]$;
  the partition $\rho$ will be finer than $\{\pi^{-1}(j)\}_j$.
  For each $i\in I$ we break $\pi^{-1}(\pi(i))$ into parts:
  we choose a minimal subset $I_i\ni i$ such that $|\sum_{j\in I_i}a_j|\geq w$;
  the set $I_i$ is included into the new partition, 
  then for every $j\in\pi^{-1}(\pi(i))\setminus I_i$ the one-element-set $\{j\}$ is also included in the partition.
  Finally the elements of $[n]\setminus \pi^{-1}(\pi(I))$ are singletons in the partition defined by $\rho$~%
  (as they were already in the one defined by $\pi$).
  Finally we put $\twt{\tuple a}{t}\xrightarrow{\rho}\twt{\tuple c}{t}$.

  Now $\rho$ is \kl{injective on} $I$ and, for every $i\in I$, we have $|c_{\rho(i)}|\geq w$.
  Additionally, by the minimality of $I_i$ and the fact that $|a_j|\leq w $ for $j\notin I$ ,
  we conclude that in fact $w\leq |c_{\rho(i)}|\leq 2w$.
  Finally for $i\notin I$ we have $|c_{\rho(i)}| = |a_i| \leq w$. 
  By the failure of $\ACondition$ in $\minion M$, some coordinate of $\tuple c$ has absolute value  $> \varepsilon\sum_j|c_j|$;
  on the other hand no coordinate of $\tuple c$ has absolute value $>2w$ and thus
  $w>\frac{\varepsilon}{2}\sum_i|c_i|$.
  This immediately gives $\sum_{i:i\in I}|c_{\rho(i)}| \geq \frac{2}{\varepsilon} w > \sum_i|c_i|$ which is a contradiction.
  The proof is finished.
\end{proof}

%

\section{A proof of \crtCref{thm:main}}
\label{sec:proof}
In this section, we provide a proof of~\Cref{thm:main}.
We start with a minor change to the way of presenting elements of $\TheMinion$.
For a tuple $\tuple a$ and a threshold $t$ we define a, potentially partial, function:
\begin{equation*}
  \funfor{\tuple a}{t}(\tuple x) = \begin{cases}
    0 & \text{ if } \braket{\tuple a | \tuple x } < t \\
    1 & \text{ if } \braket{\tuple a | \tuple x }  > t
  \end{cases}
\end{equation*}
It follows immediately from the definition that a Boolean function has an arithmetic presentation of the form $\funforleq{\tuple a}{t}$ 
if and only if it has a presentation as $\funfor{\tuple a}{s}$~%
(in some cases $t\neq s$).
This means that both presentations, disregarding partial functions, define the same minion: $\TheMinion$.

Take a minion $\minion M\subseteq\TheMinion$;
we can associate with $\minion M$ a \Cminion\ of presentations of functions in $\minion M$
using either way of presenting functions.
However, $\ACondition$ holds either in both of  \Cminions\ or in none of them.
Thus we can seamlessly switch to the new presentation,
which has some advantages over the old one. 
In particular the following proposition~%
(with an easy proof that is left as an exercise to the reader)
holds
\begin{proposition}\label{prop:approx}
  For every $\tuple a, t$ there exists $\varepsilon$ s. t. whenever\\ $\dist{\twt{\tuple a}{t}}{\twt{\tuple b}{s}}\leq\varepsilon$, 
  we have $\funfor{\tuple a}{t}(\tuple x) = \funfor{\tuple b}{s}(\tuple x)$
  for every $\tuple x$ in the intersection of domains of $\funfor{\tuple a}{t}$ and $\funfor{\tuple b}{s}$.
\end{proposition}

\subsection{Tractability of~\Cref{thm:main}}

Let $\Pol(\relstr A,\relstr B)\subseteq\TheMinion$ be a Boolean \kl{minion} satisfying $\PCondition$.
Let $\minion M$ be the \Cminion\ of presentations of functions in $\Pol(\relstr A, \relstr B)$.
Clearly $\ACondition$ holds in $\minion M$.

If $\Pol(\relstr A, \relstr B)$ contains a constant function, then $\PCSP(\relstr A,\relstr B)$ is trivial.
Otherwise $\Pol(\relstr A, \relstr B)$ splits into two \kl{minions}: functions with the unary \kl{minor} equal to $x\mapsto x$ form the one, 
and functions with the unary \kl{minor} equal to $x\mapsto 1-x$ form the second.
At least one of the two \kl{minions} satisfies $\PCondition$, and, again, the \Cminion\ of its presentations satisfies $\ACondition$.
Without loss of generality~%
(the other case is an alphabetical variant)
we assume that it's the first one~%
($x\mapsto x$) and we restrict the \Cminion\ of presentations to these functions only: $\minion M_+$.

We apply~\Cref{thm:amainEZ} to $\minion M_+$.
Note that $\minion M_+$ cannot contain a \kl(c){constant}, in the sense of~\cref{sec:Abstract};
if it did, the Boolean function it defines would have a constant unary \kl{minor}, and we excluded this possibility.

The next case to consider is that $\closure{\minion M_+}$ contains $\AMTHR{t}$ for some $0\leq t\leq 1$.
Consider $t=0$; the function $\funfor{\frac{1}{m},\dotsc,\frac{1}{m}}{0}$ is defined everywhere except on the tuple $(0,\dotsc,0)$
and evaluates to $1$ there.
By~\Cref{prop:approx} we can find $\varepsilon$ for this partial function, and by our assumption we have a member of $\minion M_+$ 
no further than $\varepsilon$ from $\agenthr{m}{0}$. 
This member defines an \kl{idempotent} total function since we are working with the \kl{idempotent} part of $\Pol(\relstr A,\relstr B)$, and thus it is an $m$-ary maximum function.
Thus, if $t=0$, the set $\Pol(\relstr A,\relstr B)$ contains max functions of all arities and the PCSP (in search version) is tractable~%
(for details see~\cref{app:Algorithms}).
If $t=1$ the situation is analogous, except that we obtain min functions of all arities and search tractability follows again.
For $0<t<1$ we obtain a threshold function for every arity $m$ such that $mt$ is not an integer.
If $t$ is rational, a well-known algorithm provides search tractability~%
(details in~\cref{app:Algorithms}); 
if $t$ is irrational, we have symmetric functions of all arties, and a decision version $\PCSP$ is tractable;
tractability in a search version requires an extra step described in~\cref{app:Algorithms}.

It cannot be so that $-\AMTHR{t}\subseteq\closure{\minion M_+}$:
the function $-\agenthr{1}{t}$ satisfies $1\mapsto 0$ if $t=0$,
$0\mapsto 1$ if $t=-1$ and both if $-1<t<0$. 
None of these evaluations are possible in the \kl{idempotent} part of $\Pol(\relstr A,\relstr B)$.

In the remaining case $-\AMAT\cup\AMAT\subseteq\closure{\minion M_+}$;
we will be using the $\AMAT$ part of the inclusion only.
It is easy to see that 
\begin{equation*}
\funfor{\underbrace{\frac{1}{m},\dotsc,\frac{1}{m}}_m, \underbrace{\frac{1}{m-1},\dotsc,\frac{1}{m-1}}_{m-1}}{0}
\end{equation*}
is defined everywhere 
apart from $(0,\dotsc,0)$ and $(1,\dotsc,1)$.
In fact, in its domain, the function equals
\begin{equation*}
  \funforleq{\underbrace{1,\dotsc,1}_{m},\underbrace{-1,\dotsc,-1}_{m-1}}{0}
\end{equation*}
and this, again by known results, provides search tractability~%
(see~\cref{app:Algorithms}).

\subsection{Hardness of~\Cref{thm:main}}

  In this section, we use the \kl{injective layered condition} to prove hardness for \kl{minions} failing $\PCondition$.
  We need to be able to reason about the coordinates of functions in $\TheMinion$ 
  independently of their presentations.
  Note that any coordinate of $f$ is \kl{monotone} or \kl{antimonotone}, or both~%
  (it suffices to inspect any presentation of $f$).
  For any fixed $f\in\TheMinion$, we will be defining a relation $\less$, a preorder on the coordinates of $f$, depending on the coordinate \kl{monotonicity}:

  \AP
  For coordinates $i$ and $j$ that are both \kl{monotone}~%
  (or both \kl{antimonotone}) we put $i\intro*\less j$ if for any $\tuple a, \tuple b$
  satisfying $a_i=1=b_j, a_j=0=b_i$
  and $a_k=b_k$ for $k\neq i,j$ we have
  $f(\tuple a)\leq f(\tuple b)$~%
  ($f(\tuple a) \geq f(\tuple b)$ respectively).
  The intuition is that $j$ ``carries at least as much (positive or negative) weight'' as $i$. 
  To define the relation for mixed pairs of coordinates, say $i$ and $j$,
  we consider a \kl{minor} $f_{ij}$ identifying $i$ and $j$, and leaving all the other coordinates untouched.
  We put $i\reintro*\less j$ if the new coordinate $\{ij\}$ in the \kl{minor} has the same \kl{monotonicity} as $j$~
  (e.g. the new coordinate is \kl{monotone} and so is $j$).
  Note that the new coordinate in the \kl{minor} can be both \kl{monotone} and \kl{antimonotone}.

  \AP
  Whenever $i\less j\less i$ we write $i\intro*\equi j$.
  If $1\equi 2$ are both \kl{monotone}, or both \kl{antimonotone}, then $f(x,y,\tuple z) = f(y,x,\tuple z)$ for all $x,y,\tuple z$.
  If $1\equi 2$, 
  and one coordinate is \kl{monotone} and the other \kl{antimonotone}, then 
  the \kl{minor} $f(x,x,\tuple z)$ does not depend on $x$.
  Finally, we write $i \intro*\sless j$ if $i\less j$ and not $i\equi j$.
  We will be using the following observation:
  \begin{observation}\label{obs:untouched}
    \sloppy
    Let $\TheMinion\ni f\xrightarrow{\pi} g$ and  $i\less j$ in $f$.
    If $|\pi^{-1}(\pi(i))| = |\pi^{-1}(\pi(j))| = 1$ then $\pi(i)\less \pi(j)$.
  \end{observation}
 \noindent  
  The observation above is left with no proof, but proof of the following proposition is a big part of our reasoning
  and can be found in~\cref{app:hard}:
  
  \begin{restatable}{proposition}{canonical}\label{prop:canonical}
    For every $f\in\TheMinion$ there exists $\tuple a, t$ such that $f=\funfor{\tuple a}{t}$ and the following are equivalent:
    \begin{enumerate}
      \item $i\sless j$, \label{it:canOne}
      \item $|a_i| < |a_j|$, \label{it:canTwo}
      \item for every $\tuple b, s$ such that $\funfor{\tuple b}{s} = f$ we have $|b_i| < |b_j|$. \label{it:canThree}
    \end{enumerate}
  \end{restatable}
  \noindent We call a presentation of $f$, denoted $\tuple a,t$ in the previous proposition, a {\em canonical} presentation.
  
  Note that $i \equi j$ if and only if $|a_i| = |a_j|$ in a canonical presentation: 
  Indeed if $i\equi j$ then neither $|a_i|<|a_j|$ nor $|a_i|>|a_j|$ and thus $|a_i| = |a_j|$.
  Conversely, if $|a_i| = |a_j|$ then, using the fact that we always have $i\less j$ or $j\less i$, we immediately conclude that $i\equi j$.
  Consequently,  
  $i\less j$ if and only if $|a_i|\leq |a_j|$ in the canonical presentation.
  This further implies that the relation $\less$ is a total preorder~%
  (reflexive, transitive, and every two elements are comparable)
  and that $\equi$ is the largest equivalence relation contained in it.
  Note that we cannot change  $\sless$ to $\less$ and $<$ to $\leq$ in~\Cref{prop:canonical} as, in~\cref{it:canThree},  we can have $|b_i|\neq |b_j|$ even if $i\equi j$.

The rest of this section is devoted to proof of the hardness case of~\cref{thm:main}.
To that end, let $\Pol(\relstr A, \relstr B) \subseteq \TheMinion$ be a Boolean \kl{minion} in which $\PCondition$ fails with $\varepsilon > 0$; for convenience we assume that $\varepsilon = \frac{1}{N}$ for large enough $N$.
Let $\minion M\subseteq\TheAMinion$ consist of presentations of functions in $\Pol(\relstr A, \relstr B)$. 
Clearly $\ACondition$ fails in $\minion M$ with $\varepsilon$.
The following claim follows immediately by considering the canonical presentation of members of $\Pol(\relstr A,\relstr B)$:
\begin{claim}\label{claim:size}
  If $f\in\Pol(\relstr A, \relstr B)$ then the greatest equivalence class of $\equi$ which consists of all maximal elements under $\less$ cannot have more than $N$ elements.
\end{claim}

We will use the \kl{injective layered condition} to establish the hardness of ~\Cref{thm:main}. 
Our \kl{choice function} will be very simple: for $f\in\Pol(\relstr A,\relstr B)$ we let $I(f)$ to be (any) set of $3N$ coordinates that are largest under $\less$.
Note that we are, in most cases, selecting both \kl{monotone} and \kl{antimonotone} coordinates.
Moreover, our selection will include some equivalence classes of $\equi$ fully and at most one class partially.

Next claim shows that if $f\xrightarrow[I(f)]{\pi} g$ in $\Pol(\relstr A,\relstr B)$ then
in all presentations of $g$ which are derived from presentations of $f$,
a ``significant'' fraction of weight lies on $\pi(I(f))$.
Note that the example of \kl{minion} $\ST$ shows that we cannot drop the ``\kl{injective on} $I(f)$'' assumption.

\begin{claim}\label{claim:propagateweight}
  Let $f \in \Pol(\relstr A, \relstr B)$ and let $\funfor{\tuple a}{t}$ be any presentation of $f$.
  If $\twt{\tuple a}{t}\xrightarrow[I(f)]{\pi}\twt{\tuple b}{t}$ then
    \begin{equation*}
      \sum_{j \in \pi(I(f))} |b_j| > \varepsilon\sum_i |b_i|.
    \end{equation*}
\end{claim}
\begin{proof}
    Let $f,\tuple a, \pi,\tuple b$ and $t$ be as in the statement.
    We split $I(f)$ into two parts: let $I'$ be a union of $\equi$ classes fully included in $I(f)$ and $I''$ be the rest of $I(f)$~%
    (i.e. a proper subset of one of the $\equi$ classes).
    Next we factor $\pi$ into two \kl{minor maps}: let $\pi'(j) = \pi(j)$ if $\pi(j)\in\pi(I')$ and $\pi'(j) = j$ otherwise i.e.
    $\pi'$ performs all identifications that are allowed so that $\pi'$ is still \kl{covered by} $I'$~%
    (we did not define a codomain of $\pi'$,i.e., the indexing set of the $\pi'$-\kl{minor} of $f$, 
    but it is inessential).
    Further, let $\pi''$ perform all remaining identifications so that 
    \begin{equation*}
      \twt{\tuple a}{t}\xrightarrow[I(f)]{\pi'}\twt{\tuple c}{t}\xrightarrow[\pi'(I(f))]{\pi''}\twt{\tuple b}{t}
    \end{equation*}
    for some $\tuple c$.
    Note that $\pi''$ is \kl{injective on} $\pi'(I')$ but moreover 
    $\pi''(j)\in \pi''(\pi'(I')) = \pi(I')$ implies that $j\in\pi'(I')$ by the construction of $\pi'$.
    By these 2 observations, for any $j\in\pi'(I')$ we have $c_j = b_{\pi''(j)}$.

    In the next step, we will show that an element of $\pi'(I')$ \kl{dominates} $\twt{\tuple c}{t}$.
    Note that if $|I'|\geq 2N$ we are in position to apply~\cref{thm:amainHARD}:
    The \kl{minor map} $\pi'$ is \kl{injective on} $I'$ and \kl{covered by} $I'$ by construction.
    Moreover $I'$ \kl{dominates} $\twt{\tuple a}{t}$ by~\cref{it:canThree} of~\cref{prop:canonical}.
    If $|I'|\geq 2N = \frac{2}{\varepsilon}$ then all assumptions of~\cref{thm:amainHARD} hold,
    and the conclusion states that a member of $\pi'(I')$ \kl{dominates} $\twt{\tuple c}{t}$.
    Therefore we can assume that $|I''| > N$.

    Let $J$ be the greatest equivalence class of $\equi$ on $\funfor{\tuple c}{t}$.
    All the \kl{dominating} coordinates of $\twt{\tuple c}{t}$ belong to $J$ by~\cref{it:canThree} of~\cref{prop:canonical}, so if $J\subseteq \pi'(I')$ we are done.
    Suppose otherwise, that is, for some $j\notin I'$ we have $\pi'(j)\in J$.
    Note that, for every $i\in I''$ we have $i\more j$ in $f$;
    moreover as $\pi'$ leaves both coordinates $i$ and $j$ untouched~%
    (i.e. $|\pi'^{-1}(\pi'(i))| = |\pi'^{-1}(\pi'(j))| = 1$)
    we have that $\pi'(i)\more\pi'(j)$ in $\funfor{\tuple c}{t}$ by~\cref{obs:untouched}.
    But then $\pi'(I'')\subseteq J$ and $|J|\geq |\pi'(I'')| = |I''|>N$ which contradicts~\cref{claim:size}.

    We established that a coordinate in $\pi'(I')$, say a coordinate $k$, \kl{dominates} $\twt{\tuple c}{t}$.
    Since $\twt{\tuple c}{t}\in\minion M$ we immediately conclude that $|c_k| > \varepsilon\sum_i |c_i|$.
    By definitions of $I', \pi'$ and $\pi''$ we get $b_{\pi''(k)} = c_k$, and since $\sum_i|c_i|\geq \sum_i|b_i|$
    we finally obtain $b_{\pi''(k)} > \varepsilon\sum_i |b_i|$ and the claim is proved.
\end{proof}

With the claim at hand, we proceed to finish our proof. The objective is to show that the \kl{injective layered condition} holds.
Consider a suitable chain of \kl{minor maps} $f_1 \xrightarrow{\pi_{1,2}} \dots \xrightarrow{\pi_{N-1,N}} f_N$ of length $N$.
Fix an arbitrary presentation of $f_1 = \funfor{\tuple a}{t}$, and consider the derived presentations for $f_i$ i.e. $\twt{\tuple a}{t}\xrightarrow{\pi_{1,i}}\twt{\tuple a^i}{t}$.

Suppose, for a contradiction, that for every $i<j$ we have $\pi_{i,j}(I(f_i))\cap I(f_j)=\emptyset$;
since $\pi_{j,N}$ is \kl{injective on} $\pi_{i,j}(I(f_i))\cup I(f_j)$ we conclude that $\pi_{i,N}(I(f_i))\cap \pi_{j,N}(I(f_j)) = \emptyset$.
Thus the family $\{\pi_{i,N}(I(f_i))\}$ contains $N$ pairwise disjoint sets.
Applying~\cref{claim:propagateweight} to $f_i\xrightarrow[I(f_i)]{\pi_{i,N}}f_N$ with $f_i$ presented as $\funfor{\tuple a^i}{t}$
we conclude that 
\begin{equation*}
  \sum_{j\in\pi_{i,N}(I(f_i))}|a^N_j| > \varepsilon\sum_j |a_j^N|.
\end{equation*}
Now we have $N$ disjoint sets, each carrying more than $\varepsilon = \frac{1}{N}$ fraction of $\sum_j |a_j^N|$ which is a contradiction.

\section{Conclusions}
We introduced a new hardness condition,
which proved essential in providing a dichotomy for Boolean \kl(PCSP){templates} with \kl{polymorphisms} in the set of linear threshold functions.
Can it be, that the new condition captures the hardness of all the Boolean intractable PCSPs?
More modestly: can the conditional result for monotone Boolean \kl{minions}~\cite{brakensiek2021conditional} be improved by using this, or another, requirement providing unconditional hardness?
The example of the linear threshold functions shows that unconditional results are possible.
On the one hand, these functions are very well-behaved from the perspective of analysis of the Boolean function. 
On the other hand, they can be a stepping stone that allows one to tackle a larger family of functions.

Finally; our hardness result relies on a reduction from the PCP theorem. 
Is there alternative reasoning proving hardness by combinatorial means, along the lines of~\cite{barto2022babypcp}?


\begin{acks}
  We thank Libor Barto for bringing the smooth version of PCP to our attention.
  We would like to thank the anonymous reviewers for the many comments that allowed us to improve the manuscript.
\end{acks}

\bibliographystyle{ACM-Reference-Format}
\bibliography{main}

\appendix

\section{The minion $\ST$}
\label{app:ST}
In this section, we focus on the equivalent characterizations of the \kl{minion} $\ST$, introduced in \cref{sec:ex-st}. Among other things, we prove that the \kl{minion} is \emph{finitely related}, i.e. is a \kl{polymorphism} \kl{minion} of a finite \kl{PCSP template}.

\AP Recall that, for a Boolean function $f$, we call $(i, j)$ a \reintro{fixing pair} of coordinates if $f(\tuple a) = 0$ whenever $a_i = 0$ and $a_j = 1$, and $f(\tuple a) = 1$ whenever $a_i = 1$ and $a_i = 0$.

\thmST*

Let us give names to the tuples that are used to generate $\ST$:
\begin{equation*}
\agenst{m} = (1,\underbrace{2,-2,4,-4,\dotsc,2^m,-2^m}_{2m})
\end{equation*}
and recall that a Boolean function $f$ is \reintro{folded} if $f(1-\tuple b) = 1-f(\tuple b)$~%
(where $1-\tuple b$ is obtained by ``flipping'' every coordinate of $\tuple b$).
\begin{observation}
    Every member of $\ST$ is \kl{folded}.
\end{observation}
\begin{proof}
    It suffices to prove the claim for the generators.
    Take any $m$ and note that $\braket{\agenst{m}|\tuple b} + \braket{\agenst{m}| 1-\tuple b} = 1$
    and thus exactly one of these integers is $\leq 0$.
    This finishes the proof.
\end{proof}
\noindent We begin with an easy implication:
\begin{proposition}[$(1) \Rightarrow (3)$]
    Every non-unary member of $\ST$ has a \kl{fixing pair}.
\end{proposition}
\begin{proof}
    Any member of $\ST$ is a \kl{minor} of one of the generators: let $\pi$ be a \kl{minor map}
    such that $\funforleq{\agenst{m}}{0}\xrightarrow{\pi} g$.

    If $\pi(2)=\dotsb =\pi(2m+1)$ then $g(\tuple x) = x_{\pi(1)}$ and, since $g$ is not unary, the \kl{fixing pair} for $g$ is $(\pi(1),?)$~%
    (where $?$ is any coordinate different from $\pi(1)$).
    Now, let $i$ be the largest number so that $\pi(2i)\neq \pi(2i+1)$ we claim that $(\pi(2i),\pi(2i+1))$ is a \kl{fixing pair} for $g$.
    Indeed, let $\tuple b$ is such that $b_{\pi(2i)} = 1$
    and $b_{\pi(2i+1)} = 0$ and let $\tuple c$ be such that $c_j = b_{\pi(j)}$ for all $j$.
    Clearly $g(\tuple b) = \funforleq{\agenst{m}}{0}(\tuple c)$ but then $\sum_{j>2i+1}^{2m+1} c_j\agenst{m}_j =0$,
    $c_{2i}\agenst{m}_{2i} + c_{2i+1}\agenst{m}_{2i+1} = 2^i$
    and $\sum_{j<2i}c_j\agenst{m}_j > -2^i$ i.e. 
    $\braket{\agenst{m}|\tuple c} > 0$ and thus $g(\tuple b) =1$. The other condition follows immediately since $g$ is \kl{folded}.
\end{proof}
\noindent Another implication is similarly simple:
\begin{proposition}[$(3) \Rightarrow (1)$]
    If $f$ is \kl{idempotent} and every non-unary \kl{minor} of $f$ has a \kl{fixing pair} then $f\in\ST$.
\end{proposition}
\begin{proof}
    The proof is an induction on arity of $f$. If the arity is $1$ then $f(x) = x = \funforleq{\agenst{0}}{0}$ so the claim holds.

    Let $f$ be of arity $n$ and let, without loss of generality, $(1,n)$
    be a \kl{fixing pair} for $f$.
    Let $g(x_1,\dotsc,x_{n-1}) = f(x_1,\dotsc,x_{n-1},x_1)$;
    by inductive assumption there exists $m$ and a \kl{minor map} $\pi$ such that $\funforleq{\agenst{m-1}}{0}\xrightarrow{\pi}g$.
    We define $\rho : [2m+1] \to [n]$
    \begin{equation*}
        \rho(i) = \begin{cases}
            1&\text{ if } i = 2m\\
            n&\text{ if } i = 2m+1\\
            \pi(i)&\text{ otherwise}
        \end{cases}
    \end{equation*}
    and claim that $\funforleq{\agenst{m}}{0}\xrightarrow{\rho} f$.
    Let $\tuple b$ be arbitrary and $c_{i} = b_{\rho(i)}$.
    If $b_1 = 1$ and $b_n=0$ then $f(\tuple b) = 1$;
    but then $c_{2m}=1$ and $c_{2m+1} = 0$ and $\funforleq{\agenst{m}}{0}(\tuple c) = 1$.
    Similarly if $b_1=0$ and $b_n=1$ then $c_{2m}=0, c_{2m+1}=1$ and 
    $f(\tuple b) = 0 = \funforleq{\agenst{m}}{0}(\tuple c)$. 
    If $b_1=b_n$ then $c_{2m} = c_{2m+1}$ and 
    \begin{equation*}
    \braket{\agenst{m}|\tuple c} = \sum_{i=1}^{2m-1} \agenst{m-1}_ib_{\pi(i)}.
    \end{equation*}
    Note that, by the choice of $\pi$, the last sum is $\leq 0$ if and only if $g(b_1,\dotsc,b_{n-1}) = 0$.
    By the definition, and the case we are working in, $g(b_1,\dotsc,b_{n-1}) = f(b_1,\dotsc,b_{n-1}, b_1) = f(\tuple b)$ and the proof is finished.
\end{proof}

\noindent\AP Recall that the \kl(PCSP){template} $(\reintro*\STl, \reintro*\STr)$ was defined as
\begin{align*}
    \intro*\STl &= (R\,\hphantom{AE_6}, \kinl{2}{4}\+\,) \\
    \intro*\STr &= (\NAE 6, \NAE 4)
\end{align*}
where
\begin{equation*}
    R =
    \relation{| Y Y | B B B |}{
        \hline
        1 & 0 & 1 & 0 & 0\\
        0 & 1 & 1 & 0 & 0\\
        1 & 0 & 0 & 1 & 0\\
        0 & 1 & 0 & 1 & 0\\
        1 & 0 & 0 & 0 & 1\\
        0 & 1 & 0 & 0 & 1\\
        \hline
    } \hspace{4em}
    \kinl{2}{4} = 
    \relation{| B B B B B B |}{
        \hline
        1 & 0 & 1 & 0 & 1 & 0\\
        1 & 0 & 0 & 1 & 0 & 1\\
        0 & 1 & 1 & 0 & 0 & 1\\
        0 & 1 & 0 & 1 & 1 & 0\\
        \hline
    }
\end{equation*}

\noindent In the remainder of this section we investigate $\Pol(\STl, \STr)$, which is not \kl{idempotent}: e.g. unary operation $x\mapsto 1-x$ belongs to $\Pol(\STl, \STr)$. 
This simple change will simplify the considerations. 

We observe that any \kl{polymorphism} of $(\STl, \STr)$ is also \kl{compatible with} a few simpler relations:

\begin{observation}
    If $f \in \Pol(\STl, \STr)$, then
    \begin{itemize}
        \item $f \in \Pol(\NAE 2, \NAE 2)$ i.e. $f$ is folded
        \item $f \in \Pol(\kinl{1}{3}, \NAE 3)$
    \end{itemize}
\end{observation}
\begin{proof}
    We argue that these relations are pp-definable in $(\STl, \STr)$. For $\NAE 2$, it suffices to notice that
    \begin{equation*}
        x \neq y \iff (xxyy) \in \kinl{2}{4} \iff (xxyy) \in \NAE 4
    \end{equation*}
    For the latter, observe that
    \begin{align*}
        (xyz) \in \kinl{1}{3} &\iff (xxyyzz) \in R \\
        (xyz) \in \NAE 3 &\iff (xxyyzz) \in \NAE 6
    \end{align*}
\end{proof}

The remainder of this appendix is devoted to the proof of the following statement:
\begin{theorem}
\label{thm:STfixing}
    The Boolean function $f$ belongs to $\Pol(\STl, \STr)$ if and only if its unary \kl{minor} is not constant, and every other \kl{minor} has a \kl{fixing pair}.
\end{theorem}
\noindent
Note that this theorem implies the $(2) \Leftrightarrow (3)$ equivalence in \cref{thm:ST}: it suffices to add both unary singleton relations to $\relstr A$ and $\relstr B$ to ensure \kl{idempotency}. 
Back to a proof of~\Cref{thm:STfixing};
one implication is much easier:
\begin{proof}[Proof of $\impliedby$]
Take a Boolean function $f$, of arity $n$, such that unary \kl{minor} is not constant and every other \kl{minor} has a \kl{fixing pair}. Take any relation $S$ in $\relstr A$ and any choice of tuples $\tuple r^1,\dotsc,\tuple r^n$ in $S$. We want to show that $f(\tuple r^1, \dots, \tuple r^n)$ belongs to the counterpart of $S$ in $\relstr B$, i.e. $\NAE{?}$.
Consider a \kl{minor} of $f$, say $g$, defined by the following $\pi$: $\pi(i)=\pi(j)$ if and only if $\tuple r^i = \tuple r^j$.

If the \kl{minor} is unary we get the conclusion immediately
(as no constant tuple belongs to a relation in $\relstr A$).

If it is not, we have a \kl{fixing pair} of coordinates for $g$;
the tuples that appear on these two coordinates are different, say $S \ni\tuple r\neq \tuple s\in S$.
The structure of relations in $\relstr A$ implies that there exist $i,j$ such that $r_i = 0 = s_j$ and $s_i = 1 = r_j$.
Since $\tuple r$ and $\tuple s$ appear on a \kl{fixing pair} of coordinates,
the resulting tuple cannot be constant.
\end{proof}

The remainder of this section is devoted to proving the other implication. We begin with a lemma that will allow us to structure the proof,
\begin{lemma}
    If $f\in\Pol(\relstr A,\relstr B)$ then every coordinate of $f$ is \kl{monotone} or \kl{antimonotone}.
\end{lemma}
\begin{proof}
    \sloppy
    Let $n$ be the arity of $f$.
    We proceed by contradiction. 
    Let a coordinate,
    without loss of generality coordinate $1$,
    be neither.
    Then there exist tuples $\tuple a ,\tuple b$ so that 
    $0=f(0,\tuple a)<f(1,\tuple a) = 1$
    and $1 = f(0,\tuple b) > f(1,\tuple b) = 0$. 
    Note that $f$ applied to $(0,1,0,1),(a_1,a_1,b_1,b_1),\dotsc,(a_n,a_n,b_n,b_n)$ produces
    \begin{equation*}
        f~
        \relation{| E E E E |}{
            0 & $a_1$ & $\dots$ & $a_n$\\
            1 & $a_1$ & $\dots$ & $a_n$\\
            0 & $b_1$ & $\dots$ & $b_n$\\
            1 & $b_1$ & $\dots$ & $b_n$\\
        }
        =
        \relation{| E |}{
            0\\
            1\\
            1\\
            0\\
        }
    \end{equation*}
    Let $\tuple a'$ and $\tuple b'$ be the tuples obtained by flipping each bit in $\tuple a$ and $\tuple b$ respectively. Since $f$ is \kl{folded}, we have that
    \begin{equation*}
        f~
        \relation{| E E E E |}{
            0 & $a_1$ & $\dots$ & $a_n$\\
            0 & $a'_1$ & $\dots$ & $a'_n$\\
            1 & $b'_1$ & $\dots$ & $b'_n$\\
            1 & $b_1$ & $\dots$ & $b_n$\\
        }
        =
        \relation{| E |}{
            0\\
            0\\
            0\\
            0\\
        } \notin \NAE 4
    \end{equation*}
    which contradicts the \kl{compatibility with} $(\kinl{2}{4}, \NAE 4)$.
\end{proof}
\noindent\AP Thus each coordinate of $f\in\Pol(\relstr A,\relstr B)$ is \kl{monotone} or \kl{antimonotone} (or both) i.e. the domain of $f$ is $\intro*\domainup{f}\cup\intro*\domaindown{f}$. 

Next we introduce a new notation: if $f$ is $n$-ary i.e. $f(x_1,\dotsc,x_n)$ then $f_{ij}$, for $i\neq j$, is the $n-1$ ary \kl{minor}
obtained from $f$ by identifying the coordinates $i$ and $j$:
\begin{equation*}
    f(x_1,\dotsc,x_{i-1},x_{ij},x_{i+1},x_{i+1},\dotsc,
    x_{j-1},x_{ij},x_{j+1},\dotsc,x_n).
\end{equation*}
The variables of the new function are $x_k$ for $k\neq i,j$
and the new variable $x_{ij}$~%
(thus we quietly part ways with indexing the variables of a function by natural numbers only).

The proof is by induction on the arity of $f\in\Pol(\relstr A,\relstr B)$. For each such $f$ we will define a \kl{fixing pair} $(i,j)$
satisfying $f(\tuple a) = 1$ if $a_i=1$ and $a_j=0$ and $f(\tuple a) = 0$ if $a_i=0$ and $a_j=1$. 
To simplify the proof we slightly abuse the notion of a \kl{fixing pair}; we
allow pairs $(i,\boxtimes)$~(and $(\boxtimes,j)$) which mean that $f(\tuple a)=1$ if $a_i=1$, and $f(\tuple a) = 0$ if $a_i = 0$~($f(\tuple a)=1$ if $a_i=0$, and $f(\tuple a) = 0$ if $a_i = 1$ respectively) i.e. $f(\tuple a) = a_i$~($f(\tuple a) = 1-a_j$ respectively).
Throughout a proof, whenever a \kl{fixing pair} is mentioned, the pair can be of this special shape.
We are ready to begin the induction. Let $f\in\Pol(\relstr A,\relstr B)$.

Only two unary functions $x\mapsto x$ and $x\mapsto 1-x$ are compatible with $(\NAE 2,\NAE 2)$; $(1,\boxtimes)$ is a \kl{fixing pair} for the first and $(\boxtimes,1)$ for the second.

Similarly the only four binary $(x_1,x_2)\mapsto x_1$,
$(x_1,x_2)\mapsto x_2$,
$(x_1,x_2)\mapsto 1-x_1$,
$(x_1,x_2)\mapsto 1-x_2$ are \kl{compatible with} $(\NAE 2,\NAE 2)$ and the \kl{fixing pairs} are $(1,\boxtimes), (2,\boxtimes), (\boxtimes,1), (\boxtimes,2)$ respectively.

Let $f$ be ternary. 
Note that $f(\tuple x)$ has a \kl{fixing pair} if and only if $1-f(\tuple x)$ has a \kl{fixing pair} and thus we can assume wlog that $f$ is \kl{idempotent}.
Consider $f$ applied to $(100,010,001)$
\begin{equation*}
    f~
    \relation{| E E E |}{
        \emc 1 & 0 & 0\\
        0 & \emc 1 & 0\\
        0 & 0 & \emc 1\\
    }
    =
    \relation{| E |}{
        $f(1, 0, 0)$\\
        $f(0, 1, 0)$\\
        $f(0, 0, 1)$\\
    }
    \in \NAE 3;
\end{equation*}
the resulting tuple cannot be constant due to \kl{compatibility with} $(\kinl{1}{3},\NAE 3)$; if it is $100$ then $(1,\boxtimes)$ is \kl{fixing} for $f$ and similarly for other tuples with a single $1$.
In the remaining case we assume, wlog, that the result is $110$ and note that, since $f$ is \kl{folded}, we have $f = \funforleq{\agenst{1}}{0}$ and $(2,3)$ is a \kl{fixing pair} for $f$.

We take a short break from the induction to prove two lemmas:
\begin{lemma}\label{lem:fixorder}
    Let $(i,j)$ be a \kl{fixing pair} for $f\in\Pol(\relstr A,\relstr B$); if 
    $i\in\domaindown{f}$~(or $j\in\domainup{f}$) then $(\boxtimes,j)$~($(i,\boxtimes)$ respectively) is also a \kl{fixing pair} for $f$.
\end{lemma}
\begin{proof}
    Let $f$ be as in the statement, and we assume without loss of generality, that 
    $(1,2)$ is a \kl{fixing pair} for $f$ and that $1\in\domaindown{f}$.
    For every $\tuple a$ we have:
    $f(1,0,\tuple a) = 1$ but also $f(0,0,\tuple a) = 1$ as coordinate $1$ is \kl{antimonotone}. 
    Similarly $f(0,1,\tuple a) = 0$ implies that  
    $f(1,1,\tuple a) = 0$. 
    This implies that $(\boxtimes,2)$ is a \kl{fixing pair} for $f$ and the first half of the lemma is proved.

    For the second half, we have $2\in\domainup{f}$, and we apply first half to the function $\tuple a\mapsto 1-f(\tuple a)$. 
    This new function has a \kl{fixing pair} $(2,1)$ and $2$ is an \kl{antimonotone} coordinate. 
    We get $(\boxtimes,1)$ to be a \kl{fixing pair} for the new function and consequently $(1,\boxtimes)$ a \kl{fixing pair} for~$f$.
\end{proof}
\begin{lemma}\label{lem:monotoneminor}
    Let $(i,j)$ be a \kl{fixing pair} for $f_{kl}$~(with $f\in\Pol(\relstr A,\relstr B))$.
    If $\{k,l\}\subseteq \domainup{f}$ and
    $i\neq\{kl\}$ then $f$ has a \kl{fixing pair}.
\end{lemma}
\begin{proof}
    Note that either of $i,j$ can be $\boxtimes$. 
    The first observation is that we can assume $j\neq \{kl\}$; indeed the coordinate $\{kl\}\in\domainup{f_{kl}}$
    and if $j=\{kl\}$ we can use~\Cref{lem:fixorder} to substitute it with $\boxtimes$.

    Thus we can assume, without loss of generality, that $i=1,j=2,k=3$ and $l=4$~(if $i$ or $j$ are $\boxtimes$ the reasoning is identical, but the coordinate vanishes).
    Take any $\tuple a$;
    $f(1,0,0,0,\tuple a) = f(1,0,1,1,\tuple a) = 1$ by the definition.
    But then $f(1,0,0,1,\tuple a) = f(1,0,1,0,\tuple a) = 1$ as coordinates $3$ and $4$ are \kl{monotone} and $f(1,0,0,0,\tuple a) = 1$. 
    Since $f$ is \kl{folded}, this shows that $(1,2)$ is a \kl{fixing pair} for $f$ as required.
\end{proof}
Going back to the inductive proof. 
We work with $f\in\Pol(\relstr A,\relstr B)$ of arity $\geq 4$. 
We additionally assume that $\domainup{f}\cap\domaindown{f} = \emptyset$;
otherwise $f$ does not depend on some $i$~%
(in fact any $i\in\domainup{f}\cap\domaindown{f}$)
and a \kl{fixing pair} for $f$ can be recovered from a \kl{fixing pair} for the \kl{minor} $f_{ij}$ for any $j\neq i$.
The proof splits into cases:

In the first case $\domaindown{f} =\emptyset$ or $\domainup{f}=\emptyset$. We assume, wlog, the first case~%
(in the second one we can substitute $f(\tuple x)$ with $1-f(\tuple x)$ and reduce to the first). 
By inductive assumption we have a \kl{fixing pair} for $f_{12}$; unless the first coordinate of the pair is $\{12\}$ we are done by~\Cref{lem:monotoneminor}; moreover the second one can be taken to be $\boxtimes$, 
by~\Cref{lem:fixorder}, as all coordinates of $f_{12}$ are \kl{monotone}.
By the same token $(\{34\},\boxtimes)$ is a \kl{fixing pair} for $f_{34}$. 
But then $f(0,0,1,1,?,\dotsc,?)$ is both $0$~(by properties of $f_{12}$) and $1$~(by $f_{34}$) which is a contradiction.
Before we proceed to another case we need a final lemma:

\begin{lemma}\label{lem:updownminor}
    Let $f\in\Pol(\relstr A,\relstr B)$ and let $k\in\domainup{f},l\in\domaindown{f}$.
    If $(i,\{kl\})$ is a \kl{fixing pair} for $f_{kl}$ then $(i,l)$ is a \kl{fixing pair} for $f$;
    if $(\{kl\},j)$ is \kl{fixing} for $f_{kl}$ then $(k,j)$ is for $f$.
\end{lemma}
\begin{proof}
    Let $f,k,l$ be as in the statement;
    we consider the case of $(i,\{kl\})$ being a \kl{fixing pair} for $f_{kl}$~%
    (the other case reduces to this one, 
    by taking $1-f(\tuple x)$ instead of $f(\tuple x)$).

    Thus $(i,\{kl\})$ is \kl{fixing} for $f_{kl}$~%
    (note that $i$ can be $\boxtimes$)
    and we assume wlog that $i=1, k=2, l=3$~%
    (if $i$ is $\boxtimes$, the first coordinate vanishes in the reasoning).
    For every $\tuple a$ we know that
    $f(1,0,0,\tuple a) = 1$, but since $k=2$ is \kl{monotone} $f(1,1,0,\tuple a) = 1$ as well.
    Since $f$ is \kl{folded}, $(1,3)$ is a \kl{fixing pair} for $f$. 
\end{proof}
Going back to the proof we consider the case of $|\domaindown{f}|=1$ or $|\domainup{f}|= 1$.
As before it suffices to deal with the first case; 
we assume wlog that $\domaindown{f} = \{1\}$.
If the coordinate $\{12\}$ appears in any \kl{fixing pair} of the \kl{minor} $f_{12}$, then we are done by \Cref{lem:updownminor}. Otherwise, choose any \kl{fixing pair} of $f_{12}$; since all remaining coordinates of the \kl{minor} are \kl{monotone}, we can assume wlog that the pair is $(3, \boxtimes)$.
We apply the same reasoning to \kl{minor} $f_{13}$
and obtain a \kl{fixing pair} $(i,\boxtimes)$.
We claim that $i=2$; otherwise~(say $i=4$)
$f(0,0,0,1,?,\dotsc,?)$ is both $0$~(by $f_{12}$) and $1$~(by $f_{13}$).
Now consider $f_{14}$~%
(note that arity of $f$ is at least $4$).
If the \kl{fixing pair} for $f_{14}$ is $(2,\boxtimes)$ then $f(0,0,1,0,?,\dotsc,?)$
is both $1$ and $0$~%
(by $f_{12}$ and $f_{14}$).
If the \kl{fixing pair} for $f_{14}$ is $(3,\boxtimes)$ the reasoning is the same, and if it contains a new coordinate, say $5$,
we have $f(0,0,1,0,0,?,\dotsc,?)$ both $1$ and $0$~(by $f_{12}$ and $f_{15}$) -- this contradiction finishes the case.

In the next case we deal with $|\domainup{f}|\geq 4$ or $|\domaindown{f}|\geq 4$;
say that $|\domainup{f}|\geq 4$.
Note that if $|\domaindown{f}|<2$ we are done by the previous cases.
Take pairwise different elements $i,j,k,l\in\domainup{f}$;
by~\Cref{lem:monotoneminor} we are done unless a \kl{fixing pair} for $f_{ij}$ is $(\{ij\},p)$ similarly for $f_{kl}$ the pair is $(\{kl\},q)$. 
If $q\neq p$ we choose any $\tuple a$ so that $a_i=a_j=a_q=0$ and $a_p=a_k=a_l=1$,
and $f(\tuple a)$ must be both $0$ and $1$ which is a contradiction.
Thus \kl{fixing pairs} for $f_{ij}$ and $f_{kl}$ share the second coordinate, say $p$.
Similarly \kl{fixing pairs} for $f_{ik}$ and $f_{jl}$ share the second coordinate, say $q$.
We claim that $p = q$; otherwise consider $f$ applied as follows
\begin{equation*}
    f~\raisebox{5.5pt}{$$
    \begin{tabular}{| E E E E E E E E E |}
        \rowcolor{gray} $i$ & $j$ & $k$ & $l$ & $p$ & $q$ & ? & \dots & ?\\ \hline
        1 & 1 & 0 & 0 & 0 & 1 & ? & \dots & ?\\
        0 & 0 & 1 & 1 & 0 & 1 & ? & \dots & ?\\
        1 & 0 & 1 & 0 & 1 & 0 & ? & \dots & ?\\
        0 & 1 & 0 & 1 & 1 & 0 & ? & \dots & ?\\
    \end{tabular}$$}
    =
    \relation{| E |}{
        1\\
        1\\
        1\\
        1\\
    }
\end{equation*}
which contradicts \kl{compatibility with} $(\kinl{2}{4}, \NAE 4)$.
Hence the second coordinates of the \kl{fixing pairs} for $f_{ij}, f_{ik},f_{jl}$ and $f_{kl}$ are the same.
Thus for every $i,j\in\domainup{f}$ the pair
$(\{ij\},1)$ is \kl{fixing} for $f_{ij}$~(where, wlog, $1$ is the common second coordinate of the \kl{fixing pairs}). 
Let $2$ be a member of $\domainup{f}$;
a \kl{fixing pair} for $f_{12}$ needs to avoid $\{12\}$ or we are done by~\Cref{lem:updownminor}; say the pair is $(3,4)$ with $3\in\domainup{f}$ and $4\in\domaindown{f}$~%
(as otherwise we could substitute $3$ or $4$ with $\boxtimes$).
Wlog $5,6\in\domainup{f}\setminus\{2,3\}$~%
(we use $|\domainup{f}|\geq 4$ here)
but then $f(0,0,0,1,1,1,?,\dotsc,?)$ is both $1$~(by the fact that $(\{56\},1)$ is \kl{fixing} for $f_{56}$) and $0$~%
(since $(3,4)$ is \kl{fixing} for $f_{12}$).

There are just a few remaining cases.
If $|\domainup{f}| = |\domaindown{f}| = 2$ we
assume wlog that $\domainup{f} = \{1,2\}$ and $\domaindown{f} = \{3,4\}$. 
The following application
\begin{equation*}
    f~
    \relation{| E E E E |}{
        0 & 1 & 0 & 1\\
        1 & 0 & 0 & 1\\
        0 & 1 & 1 & 0\\
        1 & 0 & 1 & 0\\
    }
    \in \NAE 4
\end{equation*}
implies that e.g. $f(0,1,0,1) = 1$ which,
using \kl{monotonicity} of the coordinates $1$ and $4$ and \kl{foldedness}, implies that $(2,3)$ is a \kl{fixing pair} for $f$.

If $|\domainup{f}| = 3$ and $|\domainup{f}| =2$~%
(the other $2,3$ is solved by flipping $f$)
we assume that 
$\domainup{f} = \{1,2,3\}$ and $\domaindown{f} = \{4,5\}$ and use:
\begin{equation*}
    f~
    \relation{| E E E E E |}{
        0 & 0 & 1 & 0 & 1\\
        0 & 0 & 1 & 1 & 0\\
        0 & 1 & 0 & 0 & 1\\
        0 & 1 & 0 & 1 & 0\\
        1 & 0 & 0 & 0 & 1\\
        1 & 0 & 0 & 1 & 0\\
    }
    \in \NAE 6
\end{equation*}
to conclude that e.g.
$f(0,0,1,0,1) = 1$ as before; using \kl{monotonicity} of coordinates $1,2,5$ we deduct that $(3,4)$ is a \kl{fixing pair} for $f$.

In the final case $|\domainup{f}| = |\domaindown{f}| = 3$.
Wlog we assume that $\domainup{f} =\{1,2,3\}$ and $\domaindown{f} = \{4,5,6\}$.
We have $(\{12\},i)$ is \kl{fixing} for $f_{12}$
and $(\{13\},j)$ \kl{fixing} for $f_{13}$
and $(\{23\},k)$ \kl{fixing} for $f_{23}$, or we are done by~\Cref{lem:monotoneminor}.
We can assume,
using~\Cref{lem:fixorder},
that $i,j,k$ all belong to $\domaindown{f}$~%
(if one is in $\domainup{f}$ we change it to $\boxtimes$ using~\Cref{lem:fixorder}; and if one is $\boxtimes$ we change it to $4$).
If $i,j,k$ are pairwise different we get an immediate contradiction
\begin{equation*}
    f~
    \relation{| E E E E E E |}{
        \poc 0 & \poc 0 & 1 & \nec 1 & 0 & 0\\
        \poc 0 & 1 & \poc 0 & 0 & \nec 1 & 0\\
        1 & \poc 0 & \poc 0 & 0 & 0 & \nec 1\\
    }
    =
    \relation{| E |}{
        0\\
        0\\
        0\\
    }
    \notin \NAE 3
\end{equation*}
If $i=j=k$, say equal to $4$, we consider $f_{56}$ and get a \kl{fixing pair} $(1,\{56\})$~%
(the first element of the pair is \kl{monotone}, wlog equal to $1$, and the second is $\{56\}$ or~\Cref{lem:monotoneminor} finishes the reasoning).
But then $f(0,1,1,0,1,1)$ is both $0$~(by $f_{56}$) and $1$~(by $f_{23}$).

In the final case $i=j=4$ while $k=5$ (wlog). Like in the previous case, we consider $f_{56}$ and get a \kl{fixing pair} $(?, \{56\})$ where $?$ is $1, 2$ or $3$. We claim that it has to be $1$; assume otherwise that $(2, \{56\})$ is \kl{fixing} (the reasoning for $(3, \{56\})$ is identical). But then $f(1, 0, 1, 0, 1, 1)$ is both $0$ (by $f_{56}$) and $1$ (by $f_{13}$). Applying a similar argument, we get that $(2, \{46\})$ or $(3, \{46\})$ is \kl{fixing} for $f_{46}$, say the former one. But then the following application
\begin{equation*}
    f~
    \relation{| E E E E E E |}{
        \poc 1 & 0 & \poc 1 & \nec 0 & 1 & 1\\
        0 & \poc 1 & \poc 1 & 1 & \nec 0 & 1\\
        \poc 1 & 0 & 0 & 1 & \nec 0 & \nec 0\\
        0 & \poc 1 & 0 & \nec 0 & 1 & \nec 0\\
    }
    =
    \relation{| E |}{
        1\\
        1\\
        1\\
        1\\
    }
    \notin \NAE 4
\end{equation*}
contradicts \kl{compatibility with} $(\kinl{2}{4}, \NAE 4)$.

The proof is finished.

\section{Hardness by \kl{injective layered choice}}
\label{app:smoothLC}
This section is devoted to proving that the \kl{injective layered condition} implies hardness:
\begin{theorem}
    \label{th:injlayhardness}
    Let $\relstr A, \relstr B$ be a Boolean \kl(PCSP){template}.
    If $\Pol(\relstr A,\relstr B)$ has a \kl{choice function} $I$ that satisfies the \kl{injective layered condition},
    then $\PCSP(\relstr A, \relstr B)$ is NP-hard.
\end{theorem}

The basis for the hardness will be an extension of approximate Label Cover problem, called \emph{smooth} Label Cover (introduced by Khot in \cite{khot2002smooth}), in its layered version.
The proof follows the standard reduction strategy formalized in \cite{Bible}:
first an NP-hard approximate Label Cover problem is reduced to a promise problem of satisfying a \emph{minor condition} (i.e., a finite set of minor identities) in $\Pol(\relstr A, \relstr B)$, which is in turn log-space reducible to $\PCSP(\relstr A, \relstr B)$.
The second step is proved using \emph{long code} framework (see e.g. \cite{Bible} for details).
The first step is facilitated by the \kl{choice function} $I$ which satisfies the combinatorial properties outlined in the \kl{injective layered condition}.
In the remaining part of this section we first derive the hardness of smooth layered Label Cover problem, and then prove the aforementioned first step of the reduction.

\AP
The starting point is the approximate smooth Label Cover problem in a \emph{bipartite} version.
A bipartite Label Cover \intro(LC){instance} consists of
\begin{itemize}
    \item a biregular%
    \footnote{Note that in the general statement of the LC problem, instances are not required to be regular. But it is well-known that the problem remains NP-hard even after adding this assumption (e.g. \cite{khot2002smooth}).}
    graph $(Y \cup Z, E)$;
    \item domains $[l]$ and $[r]$ where $l, r \in \mathbb{N}$;
    \item a family of constraints $\pi_{y \to z} : [l] \to [r]$, one per each edge $(y, z) \in E$.
\end{itemize}
\AP An \intro(LC){assignment} $\sigma$ maps each variable $y \in Y$ to some value in $[l]$, and each $z \in Z$ to some value in $[r]$. We say that a constraint $\pi_{y \to z}$ is \reintro(LC){satisfied} by $\sigma$ if $\pi_{y \to z}(\sigma(y)) = \sigma(z)$.

\AP
Finally an LC instance is $\delta$-\intro(LC){smooth} if for any fixed $y \in Y$ and $S \subseteq [l]$ holds%
\footnote{Khot originally defines smoothness only for sets $S$ of size 2, but it's trivial to derive the general case using Union Bound.}
\begin{equation*}
    \Pr_{z:(y, z) \in E}\left[ |\pi_{y \to z}(S)| < |S| \right] \le \delta |S|^2
\end{equation*}

We now state the hardness of approximability of smooth Label Cover problem:
\begin{theorem}[Theorem 3.5 in \cite{khot2002smooth} with $L = 1$]
\label{thm:gapsmoothLC}
    For any $\varepsilon, \delta > 0$ and sufficiently large $N$ it is NP-hard to distinguish $\delta$-\kl(LC){smooth} LC \kl(LC){instances} with domain sizes $\le N$ which are fully \kl(LC){satisfiable} from those where not even an $\varepsilon$-fraction of all constraints can be \kl(LC){satisfied}.
\end{theorem}

\AP
The reduction from this NP-hard problem will be used to derive the hardness of smooth layered Label Cover. An $L$-layered Label Cover \intro(LLC){instance} consists of
\begin{itemize}
    \item a sequence of disjoint sets $X_1, \dots, X_L$ (called layers) of variables;
    \item domains $[c_1], \dots, [c_L]$ where $c_i \in \mathbb{N}$;
    \item edges $E_{i,j} \subseteq X_i \times X_j$ between each pair of layers $i < j$; the bipartite graph induced on any $E_{i,j}$ is biregular;
    \item a family of constraints $\phi_{x_i \to x_j} : [c_i] \to [c_j]$, one per each edge $(x_i, x_j) \in E_{i,j}$.
\end{itemize}
Additionally we say that an instance is \intro(LLC){transitive} if there is an edge $(x_i, x_k) \in E_{i,k}$ whenever $(x_i, x_j) \in E_{i,j}$ and $(x_j, x_k) \in E_{j,k}$.

\AP
As before, an \intro(LLC){assignment} $\sigma$ maps each variable $x_i \in X_i$ to some value in $[c_i]$. A sequence of variables $(x_1, \dots, x_L)$ is a \intro{chain} if $x_i \in X_i$ and each pair $(x_i, x_j) \in E_{i,j}$ is an edge. A chain is said to be \reintro{weakly satisfied} if at least one constraint $\pi_{x_i \to x_j}$ in the chain is \kl(LC){satisfied}.

\AP
In order to generalize the \kl(LC){smoothness} property, we look at adjacent layers.
Formally, an $L$-layered LC \kl(LLC){instance} is $\delta$-\intro(LLC){smooth} if for any fixed $i < L$, $x_i \in X_i$ and $S \subseteq [c_i]$ holds
\begin{equation*}
    \Pr_{x_{i+1}:(x_i, x_{i+1}) \in E_{i,i+1}}\left[ |\phi_{x_i \to x_{i+1}}(S)| < |S| \right] \le \delta |S|^2
\end{equation*}

\begin{theorem}\label{thm:gapsmoothLLC}
    For any $\varepsilon, \delta > 0$, integer $L \ge 2$ and sufficiently large $N$ it is NP-hard to distinguish $\delta$-\kl(LLC){smooth} \kl(LLC){transitive} $L$-layered LC \kl(LLC){instances} with domain sizes $\le N$ which are fully \kl(LC){satisfiable} from those where not even an $\varepsilon$-fraction of all chains can be \kl{weakly satisfied}.
\end{theorem}
\begin{proof}
    We are going to extend the reduction described in \cite{brandts2021llc}. Apart from the smoothness argument, this is almost a verbatim proof.

    Put $\varepsilon' = \varepsilon / \binom{L}{2}$.
    \sloppy We reduce from a bipartite \kl(LC){instance} $\Gamma = (Y \cup Z, E, l, r, \pi)$ of $\delta$-\kl(LC){smooth} LC with a gap $\varepsilon'$ (by \cref{thm:gapsmoothLC}). Since $(Y \cup Z, E)$ is biregular, every variable $y \in Y$ has exactly $d_Y$ neighbors in $Z$, and every variable $z \in Z$ has exactly $d_Z$ neighbors in $Y$.

    We construct an \kl(LLC){instance} $\Phi$ of $L$-layered LC. Let the variable sets be $X_i = Z^i \times Y^{L - i}$ for $1 \le i \le L$ (i.e. $L$-tuples of $i$ variables from $Z$ followed by $L - i$ variables from $Y$). Let the domain sizes be $c_i = r^i \cdot l^{L - i}$ for $1 \le i \le L$. Let the edges between $X_i$ and $X_j$ be defined for pairs of tuples $\tuple x$ and $\tuple x'$ of the form
    \begin{align*}
        \tuple x &= (z_1, \dots, z_i, ~y_{i+1}, \dots, y_j, ~y_{j+1}, \dots, y_L) \in X_i \\
        \tuple x' &= (z_1, \dots, z_i, ~z_{i+1}, \dots, z_j, ~y_{j+1}, \dots, y_L) \in X_j
    \end{align*}
    such that $(y_k, z_k) \in E$ for $i < k \le j$. Let the constraint $\phi_{\tuple x \to \tuple x'}$ map $(a_1, \dots, a_L) \in [r]^i \times [l]^{L - i}$ to
    \begin{equation*}
        (a_1, \dots, a_i, ~\pi_{y_{i+1} \to z_{i+1}}(a_{i+1}), \dots, \pi_{y_j \to z_j}(a_j), ~a_{j+1}, \dots, a_L)
    \end{equation*}
    \noindent Note that \kl{chains} in the obtained \kl(LLC){instance} are in bijection with $(L - 1)$-tuples of original edges from $\Gamma$.
    Indeed, a \kl{chain} $(\tuple{x_1}, \dots, \tuple{x_L})$ is determined by $\tuple{x_1} = (z_1, y_2, \dots, y_L)$ and $\tuple{x_L} = (z_1, z_2, \dots, z_L)$ such that $(y_k, z_k) \in E$ for $1 < k \le L$.
    Moreover, for each $i < j$, every edge $(\tuple x, \tuple x') \in E_{i, j}$ appears in the same number of \kl{chains} (namely $d_Z^i \cdot d_Y^{L - j}$).
    It is a trivial observation that the obtained \kl(LLC){instance} is \kl(LLC){transitive}.
    
    To see that it is $\delta$-\kl(LLC){smooth}, fix $i < L$, $\tuple x \in X_i$, $S \subseteq [r]^i \times [l]^{L - i}$, and let $\tuple x' \in X_{i+1}$ be a random neighbor of $\tuple x$. Note that $\tuple s$ and $\phi_{\tuple x \to \tuple x'}(\tuple s)$ are identical on all coordinates except possibly $(i+1)$-th (for any $\tuple s$ in the domain of $X_i$).

    Define an equivalence $\sim$ on $S$ so that $\tuple s \sim \tuple s'$ when $s_k = s'_k$ for all $k \neq i+1$.
    Observe that $\phi_{\tuple x \to \tuple x'}(\tuple s) = \phi_{\tuple x \to \tuple x'}(\tuple s')$ iff $s \sim s'$ and $\pi_{y_{i+1} \to z_{i+1}}(s_{i+1}) = \pi_{y_{i+1} \to z_{i+1}}(s'_{i+1})$.
    Therefore $|\phi_{\tuple x \to \tuple x'}(S)| < |S|$ iff $|\phi_{\tuple x \to \tuple x'}(C)| < |C|$ for some equivalence class $C$ iff $|\pi_{y_{i+1} \to z_{i+1}}(\proj_{i+1} C)| < |\proj_{i+1} C|$. Applying Union Bound we get
    \begin{align*}
        \Pr_{\tuple x'}\left[ |\phi_{\tuple x \to \tuple x'}(S)| < |S| \right] \le\\
        \le \sum_{C \in S / \sim} \Pr_{(y_{i+1}, z_{i+1}) \in E} \left[ |\pi_{y_{i+1} \to z_{i+1}}(\proj_{i+1} C)| < |\proj_{i+1} C| \right] \le\\
        \intertext{(by smoothness of the bipartite instance $\Gamma$)}
        \le \sum_{C \in S / \sim} \delta |\proj_{i+1} C|^2 \le \sum_{C \in S / \sim} \delta |C|^2 \le \delta |S|^2
    \end{align*}
    because $|S| = \sum_{C} |C|$. This finishes the construction.

    If the original \kl(LC){instance} $\Gamma$ is fully \kl(LC){satisfiable} then so is the new one $\Phi$: indeed, if $\sigma$ is a satisfying \kl(LC){assignment} for $\Gamma$, then $\tuple x \mapsto (\sigma(x_1), \dots, \sigma(x_L))$ is a satisfying \kl(LLC){assignment} for $\Phi$.

    For soundness suppose that in $\Phi$, an \kl(LLC){assignment} $\sigma$ \kl{weakly satisfies} at least $\varepsilon$-fraction of all \kl{chains}.
    Then there exists $i < j$ such that at least $\varepsilon / \binom{L}{2} = \varepsilon'$-fraction of all \kl{chains} are \kl{weakly satisfied} at a constraint between layers $X_i$ and $X_j$.
    Every constraint between $X_i$ and $X_j$ is contained in the same number of \kl{chains}, say $C$, hence at least $\varepsilon'$-fraction of the constraints between $X_i$ and $X_j$ are \kl(LC){satisfied}.

    Choose an arbitrary coordinate $k \in [i + 1, j]$.
    Partition $X_i$ into equivalence classes such that $\tuple x, \tuple x'$ are in the same class if they are identical on all coordinates except possibly $k$-th;
    partition $X_j$ in the same way. There exists a pair of classes between which constraints exist and at least $\varepsilon'$-fraction of them are \kl(LC){satisfied}. That is, there are
    \begin{align*}
        &x_1, \dots, x_{k-1}, ~x_{k+1}, \dots, x_L \in Y \cup Z \text{ and} \\
        &x'_1, \dots, x'_{k-1}, ~x'_{k+1}, \dots, x'_L \in Y \cup Z
    \end{align*}
    such that $\sigma$ \kl(LC){satisfies} at least $\varepsilon'$-fraction of constraints between pairs of the form
    \begin{align*}
        &(x_1, \dots, x_{k-1}, ~y, ~x_{k+1}, \dots, x_L) \in X_i \\
        &(x'_1, \dots, x'_{k-1}, ~z, ~x'_{k+1}, \dots, x'_L) \in X_j
    \end{align*}
    where $(y, z) \in E$. Therefore, one can define an \kl(LC){assignment} $\sigma' : Y \cup Z \to [l] \cup [r]$ by
    letting $\sigma'(y)$ and $\sigma'(z)$ be the $k$-th element of the value resulting from applying $\sigma$ to the above tuples, respectively for $y \in Y$ and $z \in Z$.
    This \kl(LC){assignment} then \kl(LC){satisfies} at least $\varepsilon'$-fraction of all constraints $\pi$ in the original \kl(LC){instance} $\Gamma$.
\end{proof}

We can now move on to reducing smooth layered Label Cover to minor condition satisfaction. We need to briefly introduce minor conditions before that.

\AP
An $L$-layered \intro{minor condition} over $L$ disjoint sets $\funset F_i$ of function symbols is a finite set of identities of the form
\begin{equation*}
    f_i(x_1, \dots, x_n) \approx f_j(x_{\pi(1)}, \dots, x_{\pi(m)})
\end{equation*}
where $i < j$, $f_i \in \funset F_i$ and $f_j \in \funset F_j$.\AP
Such an identity is said to be \intro(MC){satisfied} in \kl{minion} $\minion M$ on $(A, B)$ if there are functions $\xi(f_i), \xi(f_j) \in \minion M$ of arities $n$ and $m$ respectively such that
\begin{equation*}
    \xi(f_i)(a_1, \dots, a_n) = \xi(f_j)(a_{\pi(1)}, \dots, a_{\pi(m)})
\end{equation*}
for all $a_1, \dots, a_n \in A$.
Moreover a minor condition $\Sigma$ is said to be \reintro(MC){satisfied} in $\minion M$ if there is an \reintro(MC){assignment} $\xi : \cup \funset F_i \to \minion M$ that satisfies all identities in $\Sigma$.\AP
We say that $\Sigma$ is \intro(MC){trivial} if it's satisfied in the \kl{minion} of \kl{projections}.   

The $L$-layered $\minion M$-gap MC is a promise problem which, given an $L$-layered \kl{minor condition},
\begin{itemize}
    \item accepts if the condition is \kl(MC){trivial}, and
    \item rejects if the condition is not \kl(MC){satisfiable} in $\minion M$.
\end{itemize}
Barto et al. \cite{Bible} show that this problem is log-space reducible to $\PCSP(\relstr A, \relstr B)$ if $\minion M = \Pol(\relstr A, \relstr B)$.

\begin{proof}[Continuation of the proof of \cref{th:injlayhardness}]
We are ready to describe the reduction.
Let $M$ be the bound on $|I(f)|$ given by the \kl{injective layered condition}. For convenience, we will use the number of layers $L = M$.
Choose any $\varepsilon, \delta > 0$ such that $\frac{1}{M^2} \ge \delta M^3 + \varepsilon$.
A $\delta$-\kl(LLC){smooth} \kl(LLC){transitive} $M$-layered LC \kl(LLC){instance} $\Phi$ is transformed into an $M$-layered \kl{minor condition} $\Sigma$ in a straightforward way, i.e. for each constraint $\pi_{x_i \to x_j}$ we introduce the identity
\begin{equation*}
    f_{x_j}(y_1, \dots, y_{r_j}) \approx f_{x_i}(y_{\pi_{x_i \to x_j}}(1), \dots, y_{\pi_{x_i \to x_j}}(r_i))
\end{equation*}
and an \kl(LLC){assignment} $\sigma$ of $\Phi$ is interpreted as an \kl(MC){assignment} $\xi : \cup \funset F_i \owns f_x \mapsto \proj_{\sigma(x)}$. By construction, any constraint in $\Phi$ is \kl(LC){satisfied} by $\sigma$ iff the corresponding identity in $\Sigma$ is \kl(MC){satisfied} by $\xi$. In particular any YES-instance of LC is transformed into a \kl(MC){trivial} \kl{minor condition}.

For soundness, suppose that $\Sigma$ is \kl(MC){satisfied} in $\minion M$ via some $\xi$. We will argue that it's possible to \kl{weakly satisfy} at least $\varepsilon$-fraction of all \kl{chains} in $\Phi$.
For convenience we will denote $\xi(f_x)$ by $f_x^{\minion M}$.

Let us recall the properties of the \kl{choice function} $I$.
Consider any \kl{chain} $(x_1, \dots, x_M)$ and the corresponding chain of \kl{minor maps} $f_{x_1} \to \dots \to f_{x_M}$; we will write $\pi_{i,j}$ instead of $\pi_{x_i \to x_j}$ for clarity.
Put $I_j^+ = I(f_{x_j}^{\minion M}) \cup \bigcup_{i < j} \pi_{i,j}(I(f_{x_i}^{\minion M}))$.
We will call the \kl{chain} $(x_1, \dots, x_M)$ \emph{injective} if each $\pi_{i,i+1}$ is \kl{injective on} $I_i^+$.
The second item of the \kl{injective layered condition} states that whenever the \kl{chain} is injective, there exist $i < j$ such that $\pi_{i,j}(I(f_{x_i}^{\minion M})) \cap I(f_{x_j}^{\minion M}) \neq \emptyset$.

Let us define an \kl(LLC){assignment} $\sigma$ by choosing a coordinate $\sigma(x) \in I(f_x^{\minion M})$ uniformly at random.
First we will argue that $\sigma$ \kl{weakly satisfies} any injective \kl{chain} in $\Phi$ with probability at least $1/M^2$; then we will show that a significant fraction of all \kl{chains} are in fact injective.

Indeed, consider any injective \kl{chain} $(x_1, \dots, x_M)$ and the corresponding chain of \kl{minor maps} $f_1 \xrightarrow{\pi_{1,2}} \dots \xrightarrow{\pi_{M-1,M}} f_M$.
By the previous discussion, there exist $i < j$ such that $\pi_{i,j}(I(f_{x_i}^{\minion M})) \cap I(f_{x_j}^{\minion M}) \neq \emptyset$. So the constraint $\pi_{i,j}$ is \kl{weakly satisfied} by $\sigma$ with probability at least $1/M^2$.

Now, to bound the fraction of injective \kl{chains}, we need to define what a ``random chain'' is. Consider the following process: pick a vertex $x_1 \in X_1$ uniformly at random, then pick its neighbor $x_2 \in X_2$ uniformly at random, then pick its neighbor and so on. We obtain a sequence $(x_1, \dots, x_L)$, which is a \kl{chain} by \kl(LLC){transitivity} of $\Phi$. Clearly any \kl{chain} can be obtained in such a way with non-zero probability.

\begin{claim}
    The distribution of the prefix $(x_1, \dots, x_i)$ of the \kl{chain} is uniform over all possible $i$-prefixes.
\end{claim}
\begin{proof}
    Let $d_i$ be the degree of every vertex in $V_i$ in the biregular graph induced by $E_{i,i+1}$. We have $|V_1| \cdot \prod_{j=1}^{M-1} d_j$ \kl{chains} in total by construction. Once $x_1$ is picked, we can extend the \kl{chain} in $\prod_{j=1}^{M-1} d_j$ ways, no matter what the choice of $x_1$ was. Similarly once the prefix $(x_1, \dots, x_i)$ is fixed, we can extend the \kl{chain} in $\prod_{j=i}^{M-1} d_j$ ways, no matter what the prefix was.
\end{proof}

Consequently, any \kl{chain} can be obtained after such a process with the same probability. So in order to argue that a significant fraction of \kl{chains} are injective, it suffices to prove that such a random chain is injective with a significant probability.

\begin{claim}
    $\Pr(\pi_{i,i+1} \text{ isn't \kl{injective on} } I_i^+) \le \delta M^4$ for each $i < M$.
\end{claim}
\begin{proof}
    When $x_1, \dots, x_i$ are fixed, the inequality follows from \kl(LLC){smoothness} since $I_i^+$ only depends on those variables:
    \begin{equation*}
        \Pr(\pi_{i,i+1} \text{ isn't \kl{injective on} } I_i^+ \mid x_1, \dots, x_i) \le \delta |I_i^+|^2 \le \delta M^4
    \end{equation*}
    where we use the fact that $|I_i^+| \le M^2$.
    It suffices to recall that any $i$-prefix of the \kl{chain} $x_1 \to \dots \to x_i$ is equally probable, so
    \begin{align*}
        \Pr(\pi_{i,i+1} \text{ isn't \kl{injective on} } I_i^+) = \\
        = \sum_{x_1, \dots, x_i} \Pr(\pi_{i,i+1} \text{ isn't \kl{injective on} } I_i^+ \mid x_1, \dots, x_i) \cdot \Pr(x_1, \dots, x_i) \le\\
        \le \delta M^4
    \end{align*}
\end{proof}

Applying Union Bound we get
\begin{align*}
    \Pr(\text{each $\pi_{i,i+1}$ is injective}) &\ge 1 - \sum_{i < M} \Pr(\pi_{i,i+1} \text{ isn't injective}) \ge \\
    &\ge 1 - \delta M^5
\end{align*}
Hence at least $(1 - \delta M^5)$-fraction of all \kl{chains} are injective.

Finally, the expected fraction of all \kl{chains} in $\Phi$ that $\sigma$ \kl{weakly satisfies} is at least
\begin{equation*}
    (1 - \delta M^5) \cdot \frac{1}{M^2} = \frac{1}{M^2} - \delta M^3 \ge \varepsilon
\end{equation*}
It follows that there exists an \kl(LLC){assignment} which \kl{weakly satisfies} at least $\varepsilon$-fraction of all \kl{chains}.

\end{proof}

\section{Hard minions}
\label{app:hard}
Recall that $i \less j$ if
\begin{itemize}
    \item $i$ and $j$ are both \kl{monotone} and for any $\tuple a, \tuple b$ satisfying $a_i = 1 = b_j$, $a_j = 0 = b_i$ and $a_k = b_k$ for $k \neq i,j$ we have $f(\tuple a) \le f(\tuple b)$, or
    \item $i$ and $j$ are both \kl{antimonotone} and for any $\tuple a, \tuple b$ satisfying $a_i = 1 = b_j$, $a_j = 0 = b_i$ and $a_k = b_k$ for $k \neq i,j$ we have $f(\tuple a) \ge f(\tuple b)$, or
    \item $i$ is \kl{monotone}, $j$ is \kl{antimonotone}, and for any $\tuple a, \tuple b$ satisfying $a_i = 1 = b_i$, $a_j = 0 = b_j$ and $a_k = b_k$ for $k \neq i,j$ we have $f(\tuple a) \le f(\tuple b)$, or
    \item $i$ is \kl{antimonotone}, $j$ is \kl{monotone}, and for any $\tuple a, \tuple b$ satisfying $a_i = 1 = b_i$, $a_j = 0 = b_j$ and $a_k = b_k$ for $k \neq i,j$ we have $f(\tuple a) \ge f(\tuple b)$.
\end{itemize}
This entire section is devoted to the proof of the following proposition:

\canonical*

\noindent We begin a proof: Fix any $f \in \TheMinion$.
The first step is to show the equivalence $(1) \iff (3)$. We start with a simple observation:
\begin{claim}\label{cl:less}
    If there exists $\tuple c, t$ such that $\funfor{\tuple c}{t} = f$ and $|c_i| \le |c_j|$, then $i \less j$.
\end{claim}
\begin{proof}
    It follows straight from the definition; we only illustrate one of the 4 cases here: assume that $i$ and $j$ are both \kl{monotone} and $c_i, c_j \ge 0$.
    Note that for any $\tuple a, \tuple b$ satisfying $a_i = 1 = b_j$, $a_j = 0 = b_i$ and $a_k = b_k$ for $k \neq i,j$ we have $f(\tuple a) \le f(\tuple b)$ as $\braket{\tuple a | \tuple c} \le \braket{\tuple b | \tuple c}$.
\end{proof}

\noindent The first consequence of this observation is the strong connectivity of the relation $\less$:
\begin{claim}\label{cl:any2coordscomparable}
    $i \less j$ or $j \less i$ for any coordinates $i, j$ of $f$.
\end{claim}
\begin{proof}
    If neither holds, then by the Observation for any $\funfor{\tuple c}{t} = f$ we have $|c_i| > |c_j|$ and $|c_i| < |c_j|$ at the same time.
\end{proof}

The strong connectivity immediately implies the so called ``trichotomy law'':
\begin{claim}\label{cl:coordstrichotomy}
    $\forall{i, j}~\text{either } i \equi j \text{ or } i \sless j \text{ or } j \sless i$.
\end{claim}
\begin{proof}
    If $i \equi j$, then we are done. So assume that $i \not\equi j$. By the previous claim it means that either $i \less j$ or $j \less i$. Consequently either $i \sless j$ or $j \sless i$ by the definition of $\sless$.
\end{proof}
The ``trichotomy law'' together with $(1) \iff (3)$ will imply that $\less$ is in fact a total \emph{preorder}, i.e. a reflexive and transitive binary relation. However before moving on to the proof of that, we need 2 more observations, which allow us to tweak the values $a_i$ and $a_j$ in any presentation $\funfor{\tuple a}{t} = f$ whenever $i \equi j$.

\begin{observation}\label{obs:opposmontweakpres}
  Let $\tuple a, t$ be such that $a_1\geq 0\geq  a_2$ 
  and the \kl{minor} $\funfor{\tuple a}{t}(x,x,x_3,\dots,x_n)$ does not depend on $x$, then:
  \begin{itemize}
    \item if $|a_1|\geq |a_2|$, then for every $0\leq\varepsilon\leq|a_1|-|a_2|$ we have
      \begin{equation*}
        \funfor{a_1-\varepsilon,a_2-\varepsilon, a_3,\dotsc,a_n}{t-\varepsilon} = \funfor{\tuple a}{t} \text{ and }
      \end{equation*}
    \item if $|a_1|\leq |a_2|$, then for every $0\leq\varepsilon\leq|a_2|-|a_1|$ we have
      \begin{equation*}
        \funfor{a_1+\varepsilon,a_2+\varepsilon, a_3,\dotsc,a_n}{t+\varepsilon} = \funfor{\tuple a}{t}.
      \end{equation*}
  \end{itemize}
\end{observation}
\begin{proof}
  We will prove the first item only~%
  (the second case is identical).
  Let $\tuple a, t, \varepsilon$ be as in the statement and let $\tuple a'$ be the new tuple.
  Let $\tuple b\in\{0,1\}^n$; if exactly one of $b_1,b_2$ is equal $1$ then:
  $\braket{\tuple a' | \tuple b} = \braket{ \tuple a | \tuple b} -\varepsilon$
  and it compares to $t-\varepsilon$ exactly as $\braket{\tuple a|\tuple b}$ compares to $t$.

  So we are left with $b_0=b_1$ and let $\tuple c = (0,0,b_3,\dotsc,b_n)$ while $\tuple d = (1,1,b_3,\dotsc,b_n)$.
  Note that by our assumption
  \begin{equation*}
    \funfor{\tuple a}{t}(\tuple c) = \funfor{\tuple a}{t}(\tuple d).
  \end{equation*}
  i.e. 
  $t\notin [\braket{\tuple a | \tuple c},\braket{\tuple a | \tuple d}]$.

  We consider the case of $\funfor{\tuple a}{t}(\tuple b) = 0$ first. 
  We immediately get $t>\braket{\tuple a | \tuple d}$ and then
  $t-\varepsilon > \braket{\tuple a | \tuple d}-\varepsilon \geq \braket{\tuple a | \tuple c} = \braket{\tuple a'| \tuple c}$ i.e. $\funfor{\tuple a'}{t-\varepsilon}(\tuple c) = 0$ as required.
  Similarly $t-\varepsilon>\braket{\tuple a | \tuple d} -2\varepsilon = \braket{\tuple a' | \tuple d}$ and $\funfor{\tuple a'}{t-\varepsilon}(\tuple d) = 0$
  as well.
  
  If $\funfor{\tuple a}{t}(\tuple b) = 1$ then $t<\braket{\tuple a | \tuple c}$
  and, clearly $t-\varepsilon <\braket{\tuple a | \tuple c} = \braket{\tuple a'|\tuple c}$ i.e. $\funfor{\tuple a'}{t-\varepsilon}(\tuple c) = 1$ as required.
  Similarly $t-\varepsilon <\braket{\tuple a| \tuple c} -\varepsilon \leq \braket{\tuple a | \tuple d } - 2\varepsilon = \braket{\tuple a' | \tuple d}$ and the proof is finished.
\end{proof}

\begin{observation}\label{obs:samemontweakpres}
  Let $\tuple a, t$ be such that $|a_1|\geq |a_2|$ 
  and 
  \begin{equation*}
    \funfor{\tuple a}{t}(x,y,x_3,\dots,x_n) = \funfor{\tuple a}{t}(y,x,x_3,\dots,x_n),
  \end{equation*}
  then:
  \begin{itemize}
    \item if both $a_1,a_2\geq 0$ then for every $0\leq\varepsilon\leq|a_1|-|a_2|$ we have
      \begin{equation*}
        \funfor{a_1-\varepsilon,a_2+\varepsilon, a_3,\dotsc,a_n}{t} = \funfor{\tuple a}{t} \text{ and }
      \end{equation*}
    \item if both $a_1,a_2\leq 0$ then for every $0\leq\varepsilon\leq|a_1|-|a_2|$ we have
      \begin{equation*}
        \funfor{a_1+\varepsilon,a_2-\varepsilon, a_3,\dotsc,a_n}{t} = \funfor{\tuple a}{t}.
      \end{equation*}
  \end{itemize}
\end{observation}
\begin{proof}
  Again, we prove only the first item.
  Let $\tuple a, t, \varepsilon$ be as in the statement and let $\tuple a'$ be the new tuple.
  Consider and arbitrary $\tuple b\in\{0,1\}^n$;
  if $b_1=b_2$ then $\braket{\tuple a | \tuple b} = \braket{\tuple a'|\tuple b}$ and thus, clearly, $\funfor{\tuple a}{t}(\tuple b) = \funfor{\tuple a'}{t}(\tuple b)$.
  Then we let $\tuple c = (b_2,b_1,b_3,\dotsc,b_n)$
  and note that $\funfor{\tuple a}{t}(\tuple b) = \funfor{\tuple a}{t}(\tuple c)$ by our assumption. 
  
  If $b_1= 0$ and $b_2 = 1$  the reasoning depends on the value $\funfor{\tuple a}{t}(\tuple b)$:
  if it is $1$ we are done immediately as, in our case, $t < \braket{\tuple a | \tuple b}\leq \braket{\tuple a'| \tuple b}$.
  If it is $0$ we have $t>\braket{\tuple a| \tuple c} \geq \braket{\tuple a | \tuple b} +\varepsilon  = \braket{\tuple a' | \tuple b}$~%
  and we are done. 

  If $b_1 = 1$ and $b_2 = 0$ the reasoning is similar:
  if $\funfor{\tuple a}{t}(\tuple b) = 0$ we are done immediately as $\braket{\tuple a | \tuple b} \geq \braket{\tuple a' | \tuple b}$.
  Otherwise $\braket{\tuple a| \tuple c} > t$ and then $\braket{\tuple a'|\tuple b}=\braket{\tuple a | \tuple b} -\varepsilon\geq \braket{\tuple a | \tuple c} >t$
  and we are done with the whole proof.
\end{proof}

We are ready to prove the $(1) \iff (3)$ equivalence of \cref{prop:canonical}, which is a contraposition of the following statement:
\begin{lemma}\label{lem:less1eq3}
    Let $f \in \TheMinion$. Then
    \begin{equation*}
        i \less j \iff \exists{\funfor{\tuple b}{s} = f} \text{ s.t. } |b_i| \le |b_j|
    \end{equation*}
\end{lemma}
\begin{proof}
    The implication from right to left follows from \cref{cl:less}. For the opposite direction, assume that $i \less j$ and fix any $\tuple b, s$ such that $\funfor{\tuple b}{s} = f$. If $|b_i| \le |b_j|$, then we are done. Otherwise $|b_i| > |b_j|$, and by \cref{cl:less} we have $j \less i$, i.e. $i \equi j$. Now depending on the \kl{monotonicity} of coordinates $i$ and $j$, we apply \cref{obs:opposmontweakpres} or \cref{obs:samemontweakpres} with $\varepsilon = |b_i| - |b_j|$, and obtain a presentation $\funfor{\tuple b'}{s'} = f$ such that $|b'_j| - |b'_i| = |b_i| - |b_j|$ and $b'_k = b_k$ for $k \neq i,j$.
\end{proof}

\begin{proof}[Proof of $(1) \iff (3)$]
    Follows from \cref{lem:less1eq3} and the trichotomy law (\cref{cl:coordstrichotomy}).
\end{proof}

As a side consequence, \cref{lem:less1eq3} implies the transitiveness of $\less$, so it is a preorder relation. Also $\equi$ is an equivalence relation, and $\sless$ defines a linear order on the equivalence classes.

The last step is to prove the $(1) \iff (2)$ equivalence of \cref{prop:canonical}:
\begin{lemma}
    Let $f \in \TheMinion$. There exists $\tuple a, t$ such that $\funfor{\tuple a}{t} = f$ and $i \sless j \iff |a_i| < |a_j|$.
\end{lemma}
\begin{proof}
    We start with any presentation $\funfor{\tuple a}{t} = f$, and repeatedly tweak it locally applying \cref{obs:opposmontweakpres,obs:samemontweakpres}. The implication from left to right is covered by $(1) \implies (3)$, so the only goal is to have $i \equi j$ imply $|a_i| = |a_j|$.

    Namely, the set of all coordinates $[n]$ splits into a number of equivalence classes wrt $\equi$. For any such equivalence class $I$ we want to achieve
    \begin{equation}\label{eq:avg}
        |a_j| = \frac{1}{|I|} \sum_I|a_i|
    \end{equation}
    for all $j \in I$. It should already be clear how we are going to apply \cref{obs:opposmontweakpres,obs:samemontweakpres}.

    For any equivalence class $I$ that doesn't satisfy \cref{eq:avg}, pick any $k, l \in I$ such that
    \begin{equation*}
       |a_k| < \frac{1}{|I|} \sum_I |a_i| < |a_l|
    \end{equation*}
    Depending on the signs of $a_k$ and $a_l$, we apply either \cref{obs:opposmontweakpres} or \cref{obs:samemontweakpres} with a suitable $\varepsilon > 0$, modifying only these 2 coefficients so that
    \begin{equation*}
       |a'_k| = \frac{1}{|I|} \sum_I |a_i|
    \end{equation*}
    Observe that $\sum_I |a_i|$ doesn't change after such an operation.
    It suffices to note that each step increases the number of coordinates that satisfy \cref{eq:avg} by one, so after at most $n$ steps we get the desired \emph{canonical} presentation.
\end{proof}

\section{Algorithms}
\label{app:Algorithms}
Ones of the most powerful search algorithms for Boolean PCSPs were introduced in \cite{brakensiek2017promise}. For arithmetic \kl{minions}, we only need a few of them, all based on a linear program relaxation:

\begin{theorem}[\cite{brakensiek2017promise,ficak2019symmetric}]\label{thm:andorEZ}
    Let $(\relstr A, \relstr B)$ be a Boolean \kl(PCSP){template}. If $\Pol(\relstr A, \relstr B)$ contains $\MAX$ or $\MIN$, then $\PCSP(\relstr A, \relstr B)$ (in the search version) is in P.
\end{theorem}

\begin{theorem}[\cite{brakensiek2017promise,ficak2019symmetric}]\label{thm:atEZ}
    Let $(\relstr A, \relstr B)$ be a Boolean \kl(PCSP){template}. If $\AT \subseteq \Pol(\relstr A, \relstr B)$, then $\PCSP(\relstr A, \relstr B)$ is in P.
\end{theorem}

\begin{theorem}[\cite{brakensiek2017promise,ficak2019symmetric}]\label{thm:thrEZ}
    Let $(\relstr A, \relstr B)$ be a Boolean \kl(PCSP){template}. If $\THR{t} \subseteq \Pol(\relstr A, \relstr B)$ for some $t$, then $\PCSP(\relstr A, \relstr B)$ is in P.
\end{theorem}

\noindent For the threshold \kl{minions} \cite{brakensiek2017promise,ficak2019symmetric} provide an algorithm only under assumption $t \in \mathbb{Q}$. The remaining part of the section is devoted to a small adjustment to the aforementioned algorithm in case $t \notin \mathbb{Q}$.

We follow the algorithm described in \cite{brakensiek2017promise} closely. It is based on a linear programming relaxation approach.

\begin{definition}
Given an instance $\relstr I$ of a Boolean $\PCSP(\relstr A, \relstr B)$,
the \emph{basic linear programming (BLP) relaxation} of $\relstr I$ is the following linear program:
the variables are $x_v$ for every $v \in \relstr I$, while the linear constraints are as follows
\begin{itemize}
    \item For each variable $v \in \relstr I$, stipulate that $0 \le x_v \le 1$
    \item For each constraint $(\tuple v, R) \in \relstr I$, stipulate that $(x_{v_1}, \dots, x_{v_k})$ is in the convex hull of the elements of $R^{\relstr A}$
\end{itemize}
Note that a rational solution to such BLP instance (or reporting its absence) can be found in polynomial time.
\end{definition}

\begin{proof}[Proof of \cref{thm:thrEZ} when $t \notin \mathbb{Q}$]
    Let $\relstr I$ be an input \kl{structure} of $\PCSP(\relstr A, \relstr B)$ over the set of variables $[n]$.
    In the search version, the goal is to find a \kl{homomorphism} $\relstr I \to \relstr B$, given that a \kl{homomorphism} $\relstr I \to \relstr A$ exists (though, it's not provided).

    The algorithm goes as follows:
    \begin{itemize}
        \item Construct the BLP relaxation;
        \item Let $1$ be any variable in $\relstr I$;
        \item Fix $x_1 = 0$, and re-solve the BLP,
        \item If no solutions, fix $x_1 = 1$, and re-solve the BLP.
    \end{itemize}

    Assuming that $\relstr I \to \relstr A$, this procedure provides a rational solution $w \in [0, 1]^n$ such that $w_1 \in \Boolean$.
    We claim that the map
    \begin{equation*}
        h_t(i) = \funfor{1}{t}(w_i) = \begin{cases}
            1, &\quad\text{if } w_i > t \\
            0, &\quad\text{if } w_i < t 
        \end{cases}
    \end{equation*}
    is a \kl{homomorphism} $\relstr I \to \relstr B$.
    
    Consider any constraint $((v_1, \dots, v_k), R)$ and let $r = |R^{\relstr A}|$. Enumerate the elements $y^1, \dots, y^r \in R^{\relstr A}$.
    Since $w$ is a rational solution to the LP, there exist $\alpha_1, \dots, \alpha_r \in \mathbb{Q}_{[0, 1]}$
    such that $\sum \alpha_i = 1$ and $(w_{v_1}, \dots, w_{v_k}) = \sum \alpha_i \cdot y^i$.
    Let $N \in \mathbb{N}$ be the common denominator of $\alpha_1, \dots, \alpha_r$.
    Since a threshold $t$ function $\genthr{t}{N}$ of arity $N$ is a polymorphism of $(\relstr A, \relstr B)$, we have that
    \begin{equation*}
        \tuple z := \genthr{t}{N}(
            \underbrace{y^1, \dots, y^1}_{\alpha_1 N},
            \underbrace{y^2, \dots, y^2}_{\alpha_2 N},
            \dots,
            \underbrace{y^r, \dots, y^r}_{\alpha_r N}) \in R^{\relstr B}
    \end{equation*}
    On the other hand, note that
    \begin{align*}
        w_{v_j} > t &\implies \sum_i \alpha_i \cdot y^i_j = \frac{\sum_i \alpha_i N \cdot y^i_j}{N} > t \implies z_{v_j} = 1 \\
        w_{v_j} < t &\implies \sum_i \alpha_i \cdot y^i_j = \frac{\sum_i \alpha_i N \cdot y^i_j}{N} < t \implies z_{v_j} = 0
    \end{align*}
    for every $1 \le j \le k$. In other words, $\tuple z = (h_t(v_1), \dots, h_t(v_k)) \in R^{\relstr B}$.

    Hence $h_t$ is indeed a \kl{homomorphism} $\relstr I \to \relstr B$.
    However, it isn't obvious how to compute the value $h_t(i)$ in polynomial time, since $t$ may be an arbitrary non-rational number.
    Observe that there exists $q \in \{0, w_1, \dots, w_n\}$ such that
    \begin{equation*}
        \forall{i} \+\+ w_i > t \iff w_i > q
    \end{equation*}
    Consequently, the functions $\funfor{1}{t}$ and $\funforleq{1}{q}$ are equal on the set $\{w_1, \dots, w_n\}$.
    The idea is now to compute the maps $h_q$ for each $q \in \{0, w_1, \dots, w_n\}$ (at most $n+1$ of them).
    At least one of them is a valid \kl{homomorphism} $\relstr I \to \relstr B$ by the discussion above.
    Since we can verify whether $h_q$ is a valid \kl{homomorphism} in a polynomial time, and there is $O(n)$ maps $h_q$ to consider,
    the running time of the algorithm remains polynomial.
\end{proof}


\end{document}